\documentclass[12pt,a4paper]{article}
\usepackage[utf8]{inputenc}
\usepackage{lmodern}
\usepackage[T1]{fontenc}
\usepackage[english]{babel}
\usepackage{ifpdf}
\usepackage[usenames,dvipsnames]{xcolor}
\usepackage{graphicx}
\usepackage[shortlabels]{enumitem}
\usepackage{hyperref}
\usepackage[numbers,sort&compress]{natbib}
\usepackage{uri}
\usepackage[font=footnotesize,width=\textwidth]{caption}
\usepackage{booktabs}
\usepackage{colortbl}
\usepackage[top=1.8cm, bottom=2.8cm]{geometry}
\usepackage[affil-it]{authblk}
\usepackage[parfill]{parskip}
\usepackage[setbuildercolon]{matmatix}

\ifpdf
  \hypersetup{
    pdfauthor={François Bienvenu}, 
    pdftitle={Revisiting Shao and Sokal's B2 index of phylogenetic balance}
  }
\fi

\hypersetup{colorlinks, allcolors = {MidnightBlue}}

\newenvironment{mathlist}
{\begin{enumerate}[label={\upshape(\roman*)}, align=left, widest=iii, leftmargin=*]}
{\end{enumerate}\ignorespacesafterend}

\makeatletter
\newtheorem*{rep@theorem}{\rep@title}
\newcommand{\newreptheorem}[2]{%
\newenvironment{rep#1}[1]{%
 \def\rep@title{#2 \ref*{##1}}%
 \begin{rep@theorem}}%
 {\end{rep@theorem}}}
\makeatother

\newtheorem{theorem}{Theorem}[section]
\newreptheorem{theorem}{Theorem}
\newtheorem{proposition}[theorem]{Proposition}
\newreptheorem{proposition}{Proposition}
\newtheorem{lemma}[theorem]{Lemma}
\newtheorem{corollary}[theorem]{Corollary}
\theoremstyle{definition}

\newtheorem{definitionX}[theorem]{Definition}
\newenvironment{definition}
  {\pushQED{\qed}\definitionX}
  {\popQED\enddefinitionX}

\newtheorem{remarkX}[theorem]{Remark}
\newenvironment{remark}
  {\pushQED{\qed}\remarkX}
  {\popQED\enddefinitionX}

\newtheorem{exampleX}[theorem]{Example}
\newenvironment{example}
  {\pushQED{\qed}\exampleX}
  {\popQED\enddefinitionX}

\newcommand\VarGiven{\nonscript\:\delimsize\vert\nonscript\:\mathopen{}}

\newcommand{\Ball}[2]{#1_{[#2]}}
\newcommand{\Cat}[1]{\mathrm{Cat}\mleft(#1\mright)}
\newcommand{\FatCat}[1]{\mathrm{FCat}\mleft(#1\mright)}
\newcommand{\lowprime}{\mkern-.5mu\raise0.4ex\hbox{$\scriptstyle\prime$}}
\newcommand{\HIGH}{\cellcolor[gray]{0.9}}
\newcommand{\NS}[1]{\textcolor{gray}{#1}}

\newcommand{\revision}{}

\title{\vspace{-1.3ex}Revisiting Shao and Sokal's $B_2$ index\linebreak of phylogenetic balance}

\begin{document}

\author[1,2]{François Bienvenu}
\author[1,3]{Gabriel Cardona}
\author[1]{Celine Scornavacca}

\affil[1]{Institut des Sciences de l'Evolution de Montpellier
(Université de Montpellier, CNRS,\newline IRD, EPHE),
F-34095 Montpellier, France}
\affil[2]{UMR AGAP
(Université de Montpellier, CIRAD, INRAE, L'institut Agro),\newline
F-34398 Montpellier, France}
\affil[3]{Department of Mathematics and Computer Science,
University of the Balearic Islands,\newline {Ctra.\,Valldemossa km\,7.5},
E-07120 Palma, Spain}

\maketitle

\begin{abstract}
Measures of phylogenetic balance, such as the Colless and Sackin indices, play
an important role in phylogenetics. Unfortunately, these indices are
specifically designed for phylogenetic trees, and do not extend naturally to
phylogenetic networks (which are increasingly used to describe reticulate
evolution). This led us to consider a lesser-known balance index, whose
definition is based on a probabilistic interpretation that is equally
applicable to trees and to networks.  This index, known as the $B_2$ index,
was first proposed by Shao and Sokal in 1990. Surprisingly, it does not seem
to have been studied mathematically since. Likewise, it is used only
sporadically in the biological literature, where it tends to be viewed as
arcane. In this paper, we study mathematical properties of $B_2$ such as its
expectation and variance under the most common models of random trees and its
extremal values over various classes of phylogenetic networks. We also assess
its relevance in biological applications, and find it to be comparable to
that of the Colless and Sackin indices.  Altogether, our results call for a
reevaluation of the status of this somewhat forgotten measure of phylogenetic
balance.
\end{abstract}

\tableofcontents

\vspace{2ex} % HARDCODED

\section{Introduction}

\subsection{Biological context} \label{secBiological}

Whether it is to compare them, to perform simple statistical tests or
to identify general trends or patterns, it is often useful for biologists
to study trees through the lens of one or more summary statistics. In
phylogenetics, many of the most prominent summary statistics aim at capturing
the same intuitive idea: that some trees look more ``symmetric'' than others.
These statistics are collectively known as \emph{balance indices}.

Among the multitude of balance indices that have been proposed over the years
(see e.g.~\cite[Chapter~33]{Felsenstein2003}), two stand out by
their historical importance -- and are still by far
the most widely used today: the Colless index and the Sackin index.
The Colless index, introduced by Colless in~\cite{Colless1982},
is specific to rooted binary trees. It is defined as
\begin{equation} \label{eqColless}
  \mathrm{Colless}(T) \;=\;
  \sum_{i \in I} \Abs{\lambda_1(i) - \lambda_2(i)}\,, 
\end{equation}
where the sum runs over the internal vertices of~$T$ and
$\Set{\lambda_1(i), \lambda_2(i)}$ is the number of leaves in each of the two
subtrees descended from~$i$.
The Sackin index -- which, contrary to what its name suggests, was
introduced by Shao and Sokal in \cite{Shao1990} -- is defined for any rooted
tree by the formula
\begin{equation} \label{eqSackin}
  \mathrm{Sackin}(T) \;=\; \sum_{\ell \in L} \delta_\ell \,, 
\end{equation}
where the sum runs over the leaves of~$T$ and
$\delta_\ell$ denotes the depth of~$\ell$ (i.e.\ the number of edges of the
path joining it to the root).
Note that although we refer to them as balance indices, the Colless and
Sackin indices actually measure the \emph{im}balance of a phylogeny:
the higher they are, the less balanced the phylogeny.

One of the problems of the Colless and Sackin indices, which was the starting
point of this work, is that there is no single, natural way to extend their
definition to phylogenetic networks. Although hardly a concern until recently,
this is bound to become a major issue as the mounting evidence of
the major roles played by phenomena such as gene transfers and hybridization
forces biologists to abandon trees in favor of
networks \cite{Huson2006, Bapteste2013}.

\pagebreak % HARDCODED

While it is relatively easy to come up with network statistics
that reduce to the Colless\,/\,Sackin index in trees, it is hard to favor
one over the other -- or even to see why they should conform with our intuition
of what ``balance'' is.
This led us to use a different approach and consider a lesser-known
-- and to some extent forgotten~-- measure of balance known as
Shao and Sokal's $B_2$ index.  Surprisingly given its very natural
interpretation and the abundance of mathematical papers studying the properties
of other balance indices \cite[to name but a few]{Rogers1994,Rogers1996,
Blum2005,Blum2006a, Coronado2020,Coronado2020b, Cardona2013}, it seems that the
mathematical properties of this balance index have never been studied before.
Meanwhile, the current consensus in the biological literature seems to be that
$B_2$ is not as useful as other balance indices; but on closer inspection this
idea can mostly be traced to a single
study~\cite{Agapow2002}. Moreover, the specific assumptions made in
that study limit the scope of its conclusions.

The aim of this paper is to fill in the current gap in
the literature around $B_2$, in particular concerning its mathematical
properties. Our main contribution is therefore a series of propositions
and theorems about $B_2$,
but we also include a statistical analysis that strongly suggests that
the biological relevance of this index may have been underestimated when
compared to that of the Colless and Sackin indices -- and thus calls for more
empirical work on the subject.

\subsection{Definition and basic properties of $B_2$}

In this section, we recall the intuition behind Shao and Sokal's $B_2$ index
and give its formal definition in the context of phylogenetic networks.
We then list some of its elementary properties.
For this, we first need to introduce some vocabulary.

\begin{definition} \label{defRootedPhylogeny}
A \emph{rooted phylogeny} is a directed acyclic graph
that has exactly one vertex with no incoming edges. This vertex is called
the \emph{root} of the phylogeny, and the vertices with no outgoing edges
are called its \emph{leaves}.
\end{definition}

An intuitive idea in order to measure the ``balance'' of a rooted phylogeny is
to send water from its root, then let that water trickle down the
edges and accumulate in the leaves: the more evenly the water ends up being
distributed among the leaves, the more balanced the phylogeny.

In order to formalize this idea, we consider a simple forward random walk
started from the root -- that is, at each step we follow one of the outgoing
edges of the current vertex, uniformly at random, until we get trapped in
a leaf.  In a finite phylogeny, the stationary distribution of this
random walk is a probability distribution $(p_\ell)_{\ell \in L}$
on the leaf-set~$L$ of the phylogeny. To quantify the uniformity
of this distribution, we compute its Shannon entropy. This gives us the
following definition, which is due to Shao and Sokal \cite{Shao1990}.

\begin{definition} \label{defB2}
The $B_2$ index of a finite rooted phylogeny $N$ is defined as
\[
  B_2(N) = - \sum_{\ell \in L} p_\ell \log_2 p_\ell \,,
\]
where the sum runs over the leaves and $p_\ell$ is the probability that
the simple forward random walk started from the root ends in $\ell$.
\end{definition}

\begin{remark} \label{remChoiceLog}
Using base-2 logarithms in this definition is more of a convention than a
mathematical necessity. However, we will see that this turns out to be
convenient when working with binary trees, which are prominent in biology.
\end{remark}

Note that, although technically valid, Definition~\ref{defB2} is not relevant
in the case of infinite rooted phylogenies.  Indeed, the random walk can then
follow infinite paths without ever reaching a leaf. As a result, the parts of
the phylogeny that do not subtend any leaf are not accounted for by this
definition (think of the infinite binary tree, whose $B_2$ index would be 0).

Although this may not seem relevant for biological applications, from a
mathematical point of view it is useful to define $B_2$ for 
infinite phylogenies (for instance, to give a meaning to its expected value under
models that can produce infinite phylogenies, such as Galton--Watson
trees; or to simplify the study of its limiting behaviour in large phylogenies).
We do so in the context of locally finite phylogenies, i.e.\ phylogenies where
every vertex has a finite degree.

\begin{definition} \label{defB2Infinite}
The $B_2$ index of a locally finite rooted phylogeny $N$ is defined as
\[
  B_2(N) = \lim_{k\to\infty} B_2(\Ball{N}{k}) \,, 
\]
where $\Ball{N}{k}$ denotes the ball of radius $k$ centered at the root in $N$,
that is, the subgraph of $N$ induced by the vertices whose distance to
the root is at most~$k$.
\end{definition}

In particular, Definition~\ref{defB2Infinite} ensures that if $(N_i)$ is a
sequence of rooted phylogenies that converges in distribution to~$N$ (in the
local topology -- see e.g.~\cite{Curien2018}), then $B_2(N_i)$ converges in
distribution to $B_2(N)$.

\begin{example}
Let $\mathrm{CB}(h)$ be the complete binary tree with height~$h$, i.e.\
the fully symmetric binary tree with $n = 2^h$ leaves. Then,
\[
  B_2(\mathrm{CB}(h)) \;=\; \log_2 n \;=\; h\,. \qedhere
\]
\end{example}

\begin{example}
Let $\Cat{n}$ be the caterpillar with $n$ leaves (sometimes also known
as the comb), i.e.\ the rooted binary tree with $n$ leaves where every
internal node has at least one child that is a leaf. Then,
\[
  B_2(\Cat{n)} \;=\; 2 - 2^{-n + 2} \,. \qedhere
\]
\end{example}

Before closing this section and listing our main results,
let us already point out some properties of
$B_2$ that follow immediately from its definition.

\begin{proposition} \label{propRangeB2}
Let $N$ be a finite rooted phylogeny with $n$ leaves. Then,
\[
  0 \;\leq\; B_2(N) \;\leq\; \log_2 n  \,.
\]
\end{proposition}

\begin{proof}
This follows immediately from the definition of $B_2$ as a Shannon entropy.
These bounds are tight, but they can be slightly improved when considering more
restricted classes of phylogenies. This is discussed in
Section~\ref{secExtrem}.
\end{proof}

\begin{proposition} \label{propB2BinaryTrees}
If $T$\! is a rooted binary tree, then letting $\delta_{\ell}$ denote the depth
of~$\ell$ (i.e.\ the number of edges of the path joining it to the root)
we have
\[
  B_2(T) \;=\; \sum_{\ell \in L} \delta_\ell \, 2^{-\delta_\ell} \,.
\]
\end{proposition}

\begin{remark}
This was already pointed out by Shao and Sokal in \cite{Shao1990}, and in
fact in subsequent works using $B_2$ this expression is almost invariably
used as its definition, without reference to its probabilistic
interpretation.
\end{remark}

\begin{proof}
To obtain this expression from Definition~\ref{defB2}, it suffices to note
that, in a binary tree, the random walk has exactly $\delta_{\ell}$
``left\,/\,right'' decisions to make in order to get to $\ell$. As a result,
$p_\ell = 2^{-\delta_\ell}$ and the proposition follows.
\end{proof}

The next propositions are simple observations which we state as formal
propositions to avoid having to re-detail them several times.
We group them here because they are of constant use
throughout the various sections of this document, and because we think they
give some useful intuition about $B_2$.
Readers who are less interested in the technical details may skip the rest of
this section and go directly to Section~\ref{secMainResults}, where we outline
our main results.

\begin{proposition} \label{propGrafting}
Let $N$ and $N'$ be two rooted phylogenies, and let $N''$ be the rooted
phylogeny obtained by grafting $N'$\! on a leaf \,$\ell^* \in N$, i.e.\ by
making the vertices of $N$ that point to $\ell^*$ point to the root
of $N'$ instead. Then,
\[
  B_2(N'') \;=\; B_2(N) \;+\; p_{\ell^*} B_2(N')\,, 
\]
where $p_{\ell^*}$ denotes the probability of reaching $\ell^*$ in~$N$.
\end{proposition}

\begin{proof}
Let $L$ and $L'$ be the respective leaf-sets of $N$ and $N'$, and 
let $p_\ell$ and $p'_\ell$ denote the probability of reaching a leaf~$\ell$
in each of these phylogenies. Then,
\begin{align*}
  B_2(N'') \;&=\;
  -\sum_{\substack{\ell \in L\\ \ell \neq \ell^*}} p_\ell \log_2 p_\ell
  \;-\sum_{\ell \in L'} p_{\ell^*} p'_\ell \log_2 (p_{\ell^*} p'_\ell) \\
  \;&=\;
  -\sum_{\ell \in L} p_\ell \log_2 p_\ell + p_{\ell^*} \log_2 p_{\ell^*}
  - p_{\ell^*} \sum_{\ell \in L'} p'_\ell \log_2 p'_\ell
  - p_{\ell^*}\log_2 p_{\ell^*} \sum_{\ell \in L'} p'_\ell  \\[1ex]
  \;&=\;
  B_2(N) \;+\; p_{\ell^*} B_2(N')\,. \qedhere
\end{align*}
\end{proof}

Let us point out two particularly useful consequences 
of Proposition~\ref{propGrafting}.

\begin{corollary} \label{corGraftingCherry}
Let $N^*$\! be the rooted phylogeny obtained by grafting a cherry
(that is, two leaves with the same parent)
on a leaf~$\ell$ of a rooted phylogeny~$N$.
Then, $B_2(N^*) = B_2(N) + p_{\ell}$.
\end{corollary}

\begin{corollary} \label{corRootSplitRecursion}
Let $N'$ and $N''$ be two rooted phylogenies, and let $N = N' \oplus N''$ 
be the rooted phylogeny obtained by creating a new root and making it
point to the roots of $N'$ and $N''$. Then,
\[
  B_2(N) \;=\; \frac{1}{2}\big(B_2(N') + B_2(N'') \big) \;+\; 1\,.
\]
\end{corollary}

\begin{proof}
Use Proposition~\ref{propGrafting} twice to graft $N'$ and $N''$ on the
leaves of the \revision{rooted} binary tree with two leaves.
\end{proof}

\enlargethispage{2ex} %HARDCODED

\begin{proposition} \label{propDeltaB2}
Let $N$ and $N'$ be two rooted phylogenies on the same leaf-set such that the
probabilities $p_\ell$ and $p'_\ell$ of reaching $\ell$
are the same in $N$ and in $N'$ for every leaf $\ell$, with the possible
exception of two fixed leaves $\ell_1$ and $\ell_2$. Then,
\[
  \mathrm{sgn}\big(B_2(N') - B_2(N)\big) \;=\;
  \mathrm{sgn}\big((p_{\ell_1} - p'_{\ell_1})(p_{\ell_1} - p'_{\ell_2})\big)\,, 
\]
where $\mathrm{sgn}(x) \in \Set{-1, 0, +1}$ denotes the sign of $x$.
\end{proposition}

\begin{remark}
Perhaps a more intuitive way to understand Proposition~\ref{propDeltaB2}  is
to note that $(p_{\ell_1} - p'_{\ell_1})(p_{\ell_1} - p'_{\ell_2})$ has the
same sign as $|p_{\ell_1} - p_{\ell_2}| - |p'_{\ell_1} - p'_{\ell_2}|$.
Thus, $B_2(N') < B_2(N)$ if and only if $p'_{\ell_1}$ and $p'_{\ell_2}$ are
more spread out than $p_{\ell_1}$ and $p_{\ell_2}$.
\end{remark}

\begin{proof}
Letting $f\colon x \mapsto -x\log x$ and
$\Delta = p'_{\ell_1} - p_{\ell_1} = p_{\ell_2} - p'_{\ell_2}$, we have
\begin{align*}
  B_2(N') - B_2(N) \;=\;&\; f(p'_{\ell_1}) + f(p'_{\ell_2})
  - f(p_{\ell_1}) - f(p_{\ell_2})\\[0.5ex]
  \;=\;&\; \big(f(p_{\ell_1} + \Delta) - f(p_{\ell_1})\big) \;-\;
  \big(f(p'_{\ell_2} + \Delta) - f(p'_{\ell_2})\big)\,, 
\end{align*}
and the proposition follows from the strict concavity of $f$
(recall that $f$ is strictly concave if and only if $(x, y) \mapsto
(f(x) - f(y)) / (x - y)$ is decreasing in $x$ and in~$y$).
\end{proof}

\subsection{Main results} \label{secMainResults}

In Section~\ref{secExtrem}, we study the range of $B_2$ over
several classes of rooted phylogenies. This basic information is particularly
relevant if one wants to compare the~$B_2$ index of phylogenies that have a
different number of leaves, or belong to different classes (e.g, comparing
reticulated and non-reticulated phylogenies). We show in
Theorem~\ref{thmRangeRTCNs} that for
every temporal tree-child network $N$ with $n$ leaves (and
in particular for every binary tree),
\[
  2 - 2^{-n + 2} \;\leq\; B_2(n) \;\leq\; \Floor{\log_2 n} \;+\; 
  \frac{n - 2^{\Floor{\log_2 n}}}{2^{\Floor{\log_2 n}}} \,.
\]
Moreover, in the special case of binary trees, we fully characterize the
trees that attain these bounds. Notably, the only binary tree that minimizes
$B_2$ is the caterpillar tree -- in agreement with the conventional idea
that the caterpillar tree should be the least balanced tree.

Although the range of $B_2$ on binary trees is more narrow
than that of other
balances indices, such as the Colless and the Sackin indices (whose range length
is asymptotically $n^2 / 2$; see \cite{Coronado2020, Fischer2018}), 
this should not give the impression that $B_2$ is a ``coarser'' measure
of balance. In fact, \revision{it is exactly the opposite}: we show in
Proposition~\ref{propDistinctValuesB2} of Appendix~\ref{appDistinctValuesB2}
that $B_2$ takes at least $2^{\Floor{n/2} - 1}$ distinct values
on the set of binary trees with $n$ leaves -- whereas the Colless and Sackin
index, being integer-valued, can take at most $\Theta(n^2)$ different values.
As a result, the average number of trees of size $n$ that have the
same $B_2$ index is exponentially smaller than the average number of
trees of size $n$ that have the same Colless\,/\,Sackin index, meaning
that $B_2$ is better able to discriminate between trees.

\begin{figure}[h!]
  \centering
  \captionsetup{width=0.70\linewidth}
  \includegraphics[width=0.55\linewidth]{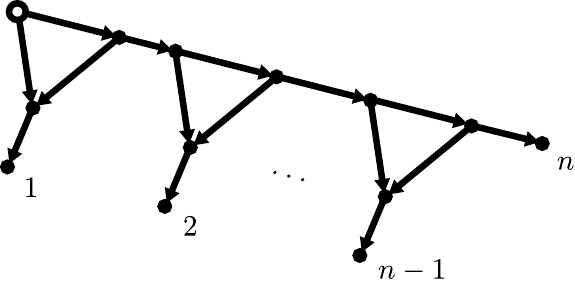}
  \caption{The fat caterpillar with $n$ leaves, $\FatCat{n}$.
  $\FatCat{n + 1}$ is obtained by grafting $\FatCat{2}$ on the $n$-th
  leaf of $\FatCat{n}$.}  
\label{figFatCat}
\end{figure}

To close Section~\ref{secExtrem}, we show in Theorem~\ref{thmMinTCN}
that the minimum of $B_2$ on the set of tree-child network is
$(\tfrac{8}{3} - \log_2(3))\mleft(1 - 4^{-n+1}\mright)$, and that
the only tree-child network that attains it is the so-called ``fat
caterpillar'', represented in Figure~\ref{figFatCat}. This result is
a relatively straightforward consequence of a series of lemmas of independent
interest that characterize the effect of various local modifications
of a rooted phylogeny on its $B_2$ index.

Section~\ref{secRandomTrees} is devoted to the study of the mean and variance
of $B_2$ for general families of random trees.
Theorem~\ref{thmExpecGW} gives an explicit, simple expression for the
expected value of the $B_2$ index of a time-inhomogeneous
Galton--Watson tree~$T_{[k]}$ stopped after $k$ generations.
In the time-homogeneous case, this expression reduces~to
\[
  \Expec{B_2(T_{[k]})} \;=\;
  \frac{\eta}{1 - \alpha} \big(1 - \alpha^k\big)\,, 
\]
with $\alpha = \Prob{Y > 0}$ and
$\eta = \Expec*{\normalsize}{\log_2(Y) \, \Indic{Y > 0}}$, where $Y$ is the
offspring distribution and $\alpha$ should be less than $1$.
In particular, for Galton--Watson trees with $2\,\mathrm{Bernoulli}(p)$
offspring distribution (that is, two offspring with probability~$p$
and~0 with probability~$1-p$), which are the most relevant Galton--Watson
trees in phylogenetics, this gives
\[
  \Expec{B_2(T_{[k]})} \;=\;
  \frac{p}{1 - p} \big(1 - p^k\big)\,.
\]
For those Galton--Watson trees it is also possible to get an explicit
expression for the variance of $B_2$ (see e.g.\ Proposition~\ref{propVarBinaryGW}).

In the second part of Section~\ref{secRandomTrees}, we study
Markov branching trees. Theorem~\ref{thmMarkovBranching} gives recurrence
relations to study the mean and variance of $B_2$ under any Markov branching
model. Applying these allows us to show in Theorems~\ref{thmYule}
and~\ref{thmUniform} that under the
ERM\,/\,Yule model with $n$ leaves,
\[
  \Expec{B_2(T^{\scriptscriptstyle \mathsf{ERM}}_n)} \;=\;
  \sum_{k = 1}^{n-1} \frac{1}{k}
  \quad\text{and}\quad
  \Var{B_2(T^{\scriptscriptstyle \mathsf{ERM}}_n)} \;
  \tendsto{n\to\infty}\; 2 - \pi^2/6
\]
and that under the PDA model with $n$ leaves,
\[
  \Expec{B_2(T^{\scriptscriptstyle \mathsf{PDA}}_n)} \;=\;
  3\,\frac{n - 1}{n + 1}
  \quad\text{and}\quad
  \Var{B_2(T^{\scriptscriptstyle \mathsf{PDA}}_n)} \;
  \tendsto{n\to\infty}\;\frac{4}{9} \,.
\]
In particular, since
$\Expec{B_2(T^{\scriptscriptstyle \mathsf{ERM}}_n)} \sim \ln n$ and
$\Expec{B_2(T^{\scriptscriptstyle \mathsf{PDA}}_n)} \sim 3$, this shows
that the ERM model generates trees that are very balanced whereas the
PDA model generates trees that are very unbalanced -- in agreement with what
we obtain using the Colless and Sackin indices \cite{Blum2006a}.

To complement this theoretical study and to evaluate the relevance of $B_2$ in
real-world applications, in
Section~\ref{secStats} we estimate the ``statistical power'' of $B_2$ when it
comes to distinguishing various types of trees from one another. We then
compare it to that of other balance indices~-- following in fact the very
approach that led Agapow and Purvis to dismiss $B_2$ as a relevant measure of
phylogenetic balance~\cite{Agapow2002}.  Our conclusions, however, are very
different: our analysis shows that the performance of each balance index varies
widely depending on the specific context in which it is used, and none of the
balance indices that we test stands out as consistently better than the others.
Neither did $B_2$ perform significantly worse than other indices:
in fact, averaged over all the scenarios that we consider (which were
selected independently of the output of our analysis),
$B_2$ happens to be the index with the best overall
performance.

Our analysis also includes a comparison of the statistical power of
pairs of balance indices. This comparison strongly supports the idea that $B_2$
better complements the Colless and Sackin indices than they complement each
other. In fact, it even suggests that, taken jointly,
the Colless and Sackin indices might the least informative pair of balance
statistics -- presumably because of their strong correlation.

A synthetic comparison of $B_2$ to other balance indices and
a discussion of its current status in phylogenetics are given in
Section~\ref{secConclusion}.

\section{Extremal values of $B_2$} \label{secExtrem}

We have seen in Proposition~\ref{propRangeB2} that, for any rooted phylogeny,
$B_2$ is non-negative and at most $\log_2 n$, where $n$ is the number of leaves
of the phylogeny. Moreover, if we consider the whole class of rooted
phylogenies, then these bounds cannot be improved, as shown by the phylogenies
depicted in Figure~\ref{figRangeB2}. But what about more restricted classes of
rooted phylogenies? In this section, we answer this question for biologically
relevant classes of rooted phylogenies: binary trees; temporal tree-child
networks; and general tree-child networks.

\vspace{1ex} %HARDCODED
\begin{figure}[h!]
  \centering
  \captionsetup{width=0.85\linewidth}
  \includegraphics[width=0.65\linewidth]{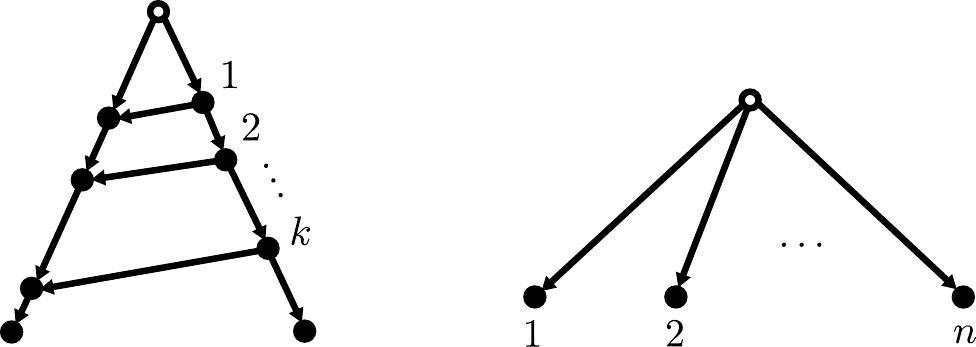}
  \caption{Examples of phylogenies showing that the bounds
  $0 \leq B_2 \leq \log_2 n$ are tight: the phylogeny on the left has
  $B_2 \to 0$ as $k \to +\infty$ and the one on the left has $B_2 = \log_2 n$.
  Note that for $n = 1$, we always have $B_2 = 0$.} 
\label{figRangeB2}
\end{figure}

Tree-child networks are a class of rooted phylogenies that were introduced
in~\cite{Cardona2009TCBB}. As of today, they are arguably the most widely
studied class of phylogenetic networks. Let us briefly recall their definition. 

\begin{definition} \label{defTCN}
A rooted phylogeny is said to be \emph{binary} if the root has outdegree~2
and every other internal vertex has either:
\begin{itemize}
\item indegree 1 and outdegree 2, in which case it is called a
  \emph{tree vertex};
\item indegree 2 and outdegree 1, in which case it is called a
  \emph{reticulation}. \qedhere
\end{itemize}
\end{definition}

\begin{definition} \label{defTCN}
A (binary) \emph{tree-child network} is a binary rooted phylogeny such
that every internal vertex has at least one child that is a tree vertex or a
leaf.
\end{definition}

Before studying the range of $B_2$ over the class of
tree-child networks, let us focus on specific subclasses that are both
particularly relevant from a biological point of view and easier to tackle from
a mathematical one.

\subsection{Binary trees and temporal tree-child networks}

Let us start with the simple case of rooted binary trees.

\begin{theorem} \label{thmExtremBinaryTrees}
Let $T$ be a rooted binary tree with $n$ leaves. Then,
\[
  2 - 2^{-n + 2} \;\leq\; B_2(T) \;\leq\; \Floor{\log_2 n} \;+\; 
  \frac{n - 2^{\Floor{\log_2 n}}}{2^{\Floor{\log_2 n}}} \,.
\]
Moreover, these bounds are sharp and
\begin{mathlist}
\item The caterpillar tree $\Cat{n}$ is the only rooted binary tree with
  $n$ leaves that minimizes $B_2$.
\item The rooted binary trees that maximize $B_2$
  are exactly the trees such that $\max \Abs{\delta_\ell - \delta_{\ell'}} \leq 1$,
  where the maximum is taken over every pair of leaves and $\delta_\ell$
  denotes the depth of leaf $\ell$.
\end{mathlist}
\end{theorem}

\begin{remark}
\revision{The caterpillar tree is also the only binary tree that maximizes
the Sackin index and the only binary tree that maximizes the Colless index.}

\revision{The binary trees that maximize $B_2$ are also exactly the
binary trees that minimize the Sackin index, see \cite{Fischer2018}.
A binary tree that minimizes the Colless index also maximizes
the $B_2$ index, but the converse is not true: there are binary trees that
maximize $B_2$ but do not minimize the Colless index. See~\cite{Coronado2020}
for more on this.}
\end{remark}

\begin{proof}[Proof of Theorem~\ref{thmExtremBinaryTrees}]
Let $T$ be a rooted binary tree. \revision{Consider a cherry of $T$ with
parent $v$ (that is, both children of $v$ are leaves) and a leaf $\ell \in T$
that is not a child of $v$}. Then,
by Corollary~\ref{corGraftingCherry}, moving the cherry from~$v$ to~$\ell$
yields a binary tree~$T'$ such that
\[
  B_2(T') - B_2(T) \;=\; 2^{-\delta_\ell} - 2^{-\delta_v} .
\]
Since it is possible to turn any binary tree into any other binary tree by
repeatedly moving cherries, it follows that:
\begin{mathlist}
\item $T$ minimizes $B_2$ if and only if it does not have a cherry with
parent~$v$ and a leaf~$\ell$ not in that cherry
such that $\delta_\ell > \delta_v$. The
caterpillar is the only such binary tree.
\item $T$ maximizes $B_2$ if and only if it does not have a cherry with
parent~$v$ and leaf~$\ell$ such that $\delta_\ell < \delta_v$, i.e.\ if and only
if the maximum difference of depth between any two leaves is at most 1.
\end{mathlist}

Finally, to compute $B_2$ for the trees that maximize it, note that if
$n$ is not a power of~2 then these trees are obtained by grafting a cherry on
$k = n - 2^{\lfloor \log_2 n\rfloor}$ of the leaves of the complete binary tree
with $2^{\lfloor \log_2 n\rfloor}$ leaves. The upper bound
then follows from Corollary~\ref{corGraftingCherry}.
\end{proof}

Let us now turn to temporal tree-child networks.
A temporal phylogeny is a phylogeny that is constrained to be compatible with
the output of a time-embedded evolutionary process. This idea is formalized
as follows.

\begin{definition} \label{defTemporal}
A rooted binary phylogeny is \emph{temporal} if there exists a function~$t$ on
its vertex set such that, for every edge $\vec{uv}$,
if $v$ is a reticulation then $t(u) = t(v)$; otherwise, $t(u) < t(v)$.
This function~$t$ is then known as a \emph{temporal labeling}.
\end{definition}

\begin{remark}
An alternative, perhaps more intuitive way to define temporal tree-child
networks is through the notion of \emph{ranked tree-child network},
or RTCNs~\cite{Bienvenu2020}.
As explained in Figure~\ref{figConstructionRTCNs}, 
RTCNs are the phylogenies generated by
sequentially grafting cherries on leaves and tridents on pairs
of leaves, keeping track of the step of the construction at which
each internal vertex was added and making it an integral part of the resulting
object.
Discarding this information and keeping only the underlying rooted phylogeny
always yields a temporal tree-child network. Moreover, every temporal
tree-child network is the underlying rooted phylogeny of some RTCN.
As a result, temporal tree-child networks and RTCNs are interchangeable
for most purposes, and one can think about them in terms of the
diagrams represented in Figure~\ref{figConstructionRTCNs}. 
\end{remark}

\vspace{1ex} %HARDCODED

\begin{figure}[h!]
  \centering
  \captionsetup{width=0.85\linewidth}
  \includegraphics[width=0.85\linewidth]{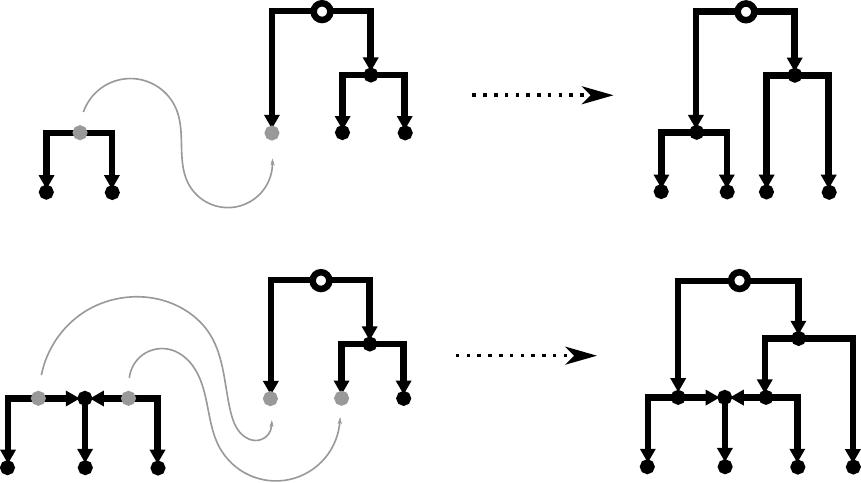}
  \caption{The operations of grafting a cherry (top) and grafting a trident (bottom).
  Sequentially performing these operations and keeping track of all the information
  of the construction, as done here through the vertical layout
  of the vertices, yields RTCNs. Temporal tree-child networks are
  obtained by discarding this vertical layout.}  
\label{figConstructionRTCNs}
\end{figure}

As it turns out, the range of $B_2$ is the same in \revision{binary trees} and
in temporal tree-child networks, as the next theorem shows.
The difference with binary trees, however, is that the caterpillar (resp.\ the binary
trees such that $\max \Abs{\delta_\ell - \delta_{\ell'}} \leq 1$) are not
the only temporal tree-child networks that minimize (resp.\ maximize)
$B_2$, as shown by the examples given in Figure~\ref{figExtremalRTCNs}.

\vspace{1ex} %HARDCODED

\begin{figure}[h!]
  \centering
  \captionsetup{width=0.75\linewidth}
  \includegraphics[width=0.65\linewidth]{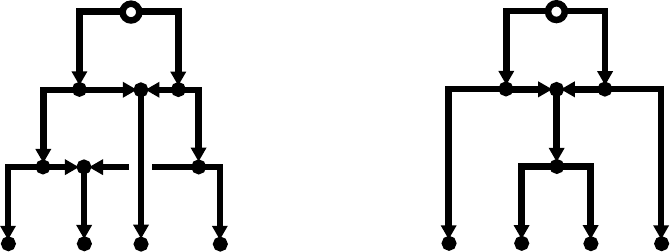}
  \caption{Examples of temporal tree-child networks that minimize (left) or
  maximize (right) $B_2$ despite being different from the trees described in
  Theorem~\ref{thmExtremBinaryTrees}.}  
\label{figExtremalRTCNs}
\end{figure}

\begin{theorem} \label{thmRangeRTCNs}
For every temporal tree-child network with $n$ leaves,
\[
  2 - 2^{-n + 2} \;\leq\; B_2(T) \;\leq\; \Floor{\log_2 n} \;+\; 
  \frac{n - 2^{\Floor{\log_2 n}}}{2^{\Floor{\log_2 n}}} \,.
\]
\end{theorem}

\begin{proof}
We will show that for any temporal tree-child network~$N$
it is possible to find two binary trees~$T'$ and $T''$ with the same number
of leaves as $N$ and such that $B_2(T') \leq B_2(N) \leq B_2(T'')$.

Assume that $N$ is not a tree, and let then $r$ be a reticulation with maximal
temporal labeling. Denote the siblings of $r$ \revision{(that is, the two
other children of each of its two parents)} by $u$ and $v$, and its child by~$w$.
Note that, since no reticulation has a greater temporal labeling than~$r$,
the phylogenies $N_u$, $N_v$ and $N_w$ subtended by $u$, $v$ and $w$ do
not contain any reticulations -- and therefore are disjoint.
This situation is represented in Figure~\ref{figProofRangeRTCNs}.
Let us show that it is possible, by removing~$r$, 
to obtain two temporal tree-child networks $N'$ and $N''$ such that
$B_2(N') \leq B_2(N) \leq B_2(N'')$.

\vspace{1ex} %HARDCODED

\begin{figure}[h!]
  \centering
  \captionsetup{width=0.95\linewidth}
  \includegraphics[width=0.95\linewidth]{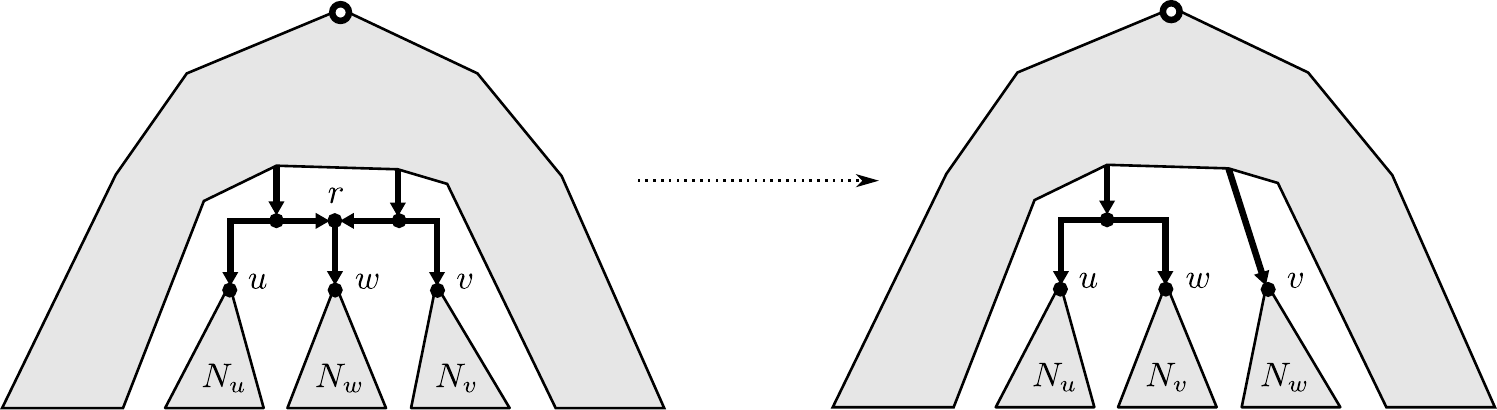}
  \caption{Left, the configuration described in the proof of
  Theorem~\ref{thmRangeRTCNs}, where the (possibly empty)
  sub-phylogenies $N_u$, $N_v$ and $N_w$ are cut from each other and from the
  rest of the phylogeny. Right, the result of the transformation that consists in
  removing the reticulation and swapping $N_v$ and $N_w$.}  
\label{figProofRangeRTCNs}
\end{figure}

Let us start with $N'$. Let $p_u$, $p_v$ and
$p_w = p_u + p_v$ be the probabilities that the simple forward random walk
goes through $u$, $v$ and $w$, respectively, and assume that $p_u \leq p_v$.
Also assume without loss of generality that $B_2(N_w) \leq B_2(N_v)$. Indeed,
if that is not the case, then swap $N_v$ and $N_w$, and let~$\tilde{N}$ be the
resulting temporal tree-child network. By Proposition~\ref{propGrafting}, 
\[
  B_2(\tilde{N}) - B_2(N) \;=\; (p_w - p_v) (B_2(N_v) - B_2(N_w)) \;<\; 0\,,
\]
and we can carry on with $\tilde{N}$ instead of~$N$.
Now, remove the edge going from the parent of~$v$ to the parent of~$w$, and
merge~$v$ and its parent; then swap $N_v$ and~$N_w$, and let $N'$ be the
resulting temporal tree-child network. This transformation is depicted in
Figure~\ref{figProofRangeRTCNs}. It can be seen as a succession of three
steps: ungrafting $N_v$ and~$N_w$; removing the reticulation; and
regrafting $N_v$ and~$N_w$.
As a result,
\[
  B_2(N') - B_2(N) \;=\; \Delta_{\mathrm{ungraft}}
  \;+\; \Delta_{\mathrm{remove}}
  \;+\; \Delta_{\mathrm{regraft}}\,, 
\]
where, by Proposition~\ref{propGrafting},
\[
  \Delta_{\mathrm{ungraft}} \;+\; \Delta_{\mathrm{regraft}} \;=\;
  (p_w - 2p_v) (B_2(N_v) - B_2(N_w)) \;\leq\; 0
\]
and, by Proposition~\ref{propDeltaB2},
$\Delta_{\mathrm{remove}} \leq 0$ since $(p_w - p_v)(p_w - 2p_v) \leq 0$.

To obtain $N''$, we use the same transformation but swapping the roles
of $u$ and~$v$. Still
with $p_u \leq p_v$, this time we can assume that $B_2(N_w) \leq B_2(N_u)$
before the transformation -- since, otherwise, swapping $N_u$ and~$N_w$ would
increase $B_2$. Replacing $v$ by $u$ in the expressions above, we see that
this time we get a temporal tree-child network $N''$ such that $B_2(N'') \geq B_2(N)$.

Finally, to obtain two binary trees~$T'$ and $T''$ such that
$B_2(T') \leq B_2(N) \leq B_2(T'')$, it suffices to apply the transformation
described above repeatedly, until all reticulations have been removed.
This concludes the proof.
\end{proof}

%%% OLD PROOF FOR THE MINIMUM ONLY (GREEDY ALGORITHM)
%\[
%  B_2(N \diamond \{\ell\}) - B_2(N) \;=\; p_\ell
%\]
%\begin{align*}
%  & B_2(N \diamond \{\ell, \ell'\}) - B_2(N) \\
%  =\;&\; p_\ell \log_2 p_\ell \;+\; p_{\ell'} \log_2 p_{\ell'}
%  \;-\; \tfrac{p_\ell}{2} \log_2 \tfrac{p_\ell}{2}
%  \;-\; \tfrac{p_{\ell'}}{2} \log_2 \tfrac{p_{\ell'}}{2}
%  \;-\; \tfrac{p_\ell + p_{\ell'}}{2} \log_2 \tfrac{p_\ell + p_{\ell'}}{2} \\[0.5ex]
%  \;=\;&\; p_\ell \mleft(1 - \tfrac{1}{2} \log_2\mleft(1 + \tfrac{p_{\ell'}}{p_\ell}\mright)\mright)
%   + p_{\ell'} \mleft(1 - \tfrac{1}{2} \log_2\mleft(1 + \tfrac{p_\ell}{p_{\ell'}}\mright)\mright)
%\end{align*}
%\[
%  B_2(N \diamond \{\ell\}) - B_2(N \diamond \{\ell, \ell'\}) \;<\;
%  \tfrac{1}{2} (p_\ell - p_{\ell'}) \;<\; 0
%\]
%\begin{align*}
%  \mathcal{R}_n \;=\; 
%  &\;\Set{N \diamond \{\ell\} \suchthat N \in \mathcal{R}_{n - 1},\, \ell \in
%  \mathrm{Leaves}(N)} \\
%  \cup\; &\;\Set{N \diamond \{\ell, \ell'\} \suchthat N \in \mathcal{R}_{n - 1},\, \ell, \ell' \in
%  \mathrm{Leaves}(N)}
%\end{align*}
%\begin{align*}
%  \min\Set{B_2(N) \suchthat N \in \mathcal{R}_n} \;=\;
%  \min\Set{B_2(N \diamond \{\ell\}) \suchthat N \in \mathcal{R}_{n - 1},\,
%  \ell \in \mathrm{Leaves}(N)} \\
%  \leq \min\Set{B_2(N) \suchthat N \in \mathcal{R}_{n - 1}} +
%  \min\Set{p_\ell \suchthat N \in \mathcal{R}_{n - 1},\,
%  \ell \in \mathrm{Leaves}(N)}
%\end{align*}

Let us now see what happens when we relax the temporal constraint and
consider the whole class of tree-child networks.

\subsection{General tree-child networks}

\begin{theorem} \label{thmMinTCN}
Let $N$ be a tree-child network with $n$ leaves. Then,
\[
  B_2(N) \;\geq\; \mleft(\tfrac{8}{3} - \log_2(3)\mright)\mleft(1 - 4^{-n+1}\mright) \,.
\]
Moreover, this bound is sharp and the only tree-child network that minimizes
$B_2$ is the so-called ``fat caterpillar'' represented in Figure~\ref{figFatCat}.
\end{theorem}

The first ingredient in our proof of Theorem~\ref{thmMinTCN} is the
following observation about rooted phylogenies.

\begin{lemma} \label{lemmaAddEdge}
Let $N$ be a rooted phylogeny and let $u$, $v$ and $w$ be three vertices of~$N$\!
such that $u$ is a parent of $v$\! and $w$,
and neither of these two vertices is an ancestor of the other.
Denote by $N'$ and $N''$ the rooted phylogenies obtained by adding an edge
between $\vec{uv}$ and $\vec{uw}$, in one direction for $N'$ and in the
other one for $N''$, as shown below.
\begin{center}
  \includegraphics[width=0.75\linewidth]{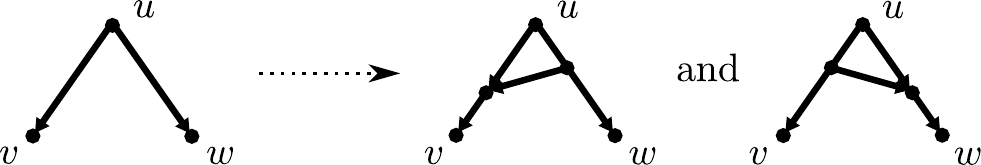}
\end{center}
Then,
\[
  B_2(N') \;+\; B_2(N'') \;\leq\; 2\, B_2(N) \,.\\[1ex]
\]
In particular, $\min\Set{B_2(N'), B_2(N'')} \leq  B_2(N)$.
Moreover, these inequalities are strict if and only if
there exists a leaf~$\ell$ such that $P_v(\ell) \neq P_w(\ell)$, where
$P_x(\ell)$ denote the probability that the simple
random walk started from $x$ ends in~$\ell$.
\end{lemma}

\begin{remark}
If $N$ is a tree-child network, then for every vertex~$x$ there exists a
leaf~$\ell$ such that every internal vertex of the path from $x$ to $\ell$ is a
tree vertex. Letting $v$ and $w$ be the vertices of the statement of the
proposition, there is thus always at least one leaf subtended by $v$ that
cannot be reached from $w$ (and vice-versa). As a result, 
the inequalities of Lemma~\ref{lemmaAddEdge} are always strict.
Note however that for $N'$ and $N''$ to be tree-child networks,
$v$ and $w$ should not be reticulations.
\end{remark}

\begin{proof}
Let $N'$ be the phylogeny obtained by making the new edge point towards
the incoming edge of $v$. For any vertex $x$ and any leaf~$\ell$,
let $p_x$ denote the probability that the simple random walk started from the
root of $N$ goes through $x$ and
$Q(\ell)$ the probability that it ends in $\ell$ without going through $v$
or~$w$. Finally, let~$P_x(\ell)$ denote the probability that the simple
random walk started from $x$ ends in~$\ell$.
Since neither of $v$ and $w$ is an ancestor of the other, the random walk
either goes through~$v$, or it goes through~$w$, or it avoids both of them.
Therefore, the probability of reaching $\ell$ in $N$ is
\[
  p_\ell \;=\; p_v P_v(\ell) \;+\; p_w P_w(\ell) \;+\; Q(\ell) \,.
\]
Similarly, the probabilities of reaching $\ell$ in $N'$ and in $N''$ are
respectively
\[
  p'_\ell \;=\; (p_v + p_u/4) P_v(\ell) \;+\; (p_w - p_u/4)
  P_w(\ell) \;+\; Q(\ell)
\]
and
\[
  p''_\ell \;=\; (p_v - p_u/4) P_v(\ell) \;+\; (p_w + p_u/4)
  P_w(\ell) \;+\; Q(\ell) \,.
\]
As a result, we have
\[
  p'_\ell \;+\; p''_\ell \;=\; 2\, p_\ell \,, 
\]
and it follows from the strict concavity of $f: x \mapsto -x \log x$ that
\[
  f(p'_\ell) \;+\; f(p''_\ell) \;\leq\;
  2\,f(p_\ell) \,,
\]
where the inequality is strict if and only if $p'_\ell \neq p''_\ell$, i.e.\ if
and only if $P_v(\ell) \neq P_w(\ell)$.
Summing these inequalities over the leaves, we thus get
\[
  B_2(N') \;+\; B_2(N'') \;\leq\; 2\, B_2(N) \,,
\]
with a strict inequality if and only
there exist a leaf~$\ell$ such that
$P_v(\ell) \neq P_w(\ell)$. This concludes the proof.
\end{proof}

Before giving the second ingredient of our proof of Theorem~\ref{thmMinTCN},
let us recall a standard fact about tree-child networks. We also recall its
proof for the sake of completeness.

\begin{lemma} \label{lemmaSaturatedTCN}
A tree-child network with $n$ leaves has at most $n - 1$ reticulations.
If it does not have $n - 1$ reticulations, then it has a vertex with two
non-reticulation children.
\end{lemma}

\begin{proof}
By the hand-shaking lemma, every rooted binary phylogeny satisfies
\[
  r + n - 1 = t\,, 
\]
where $r$ is the number of reticulations, $n$ the number of leaves and
$t$ the number of tree vertices, including the root.
In the case of a tree-child network, since every reticulation has two
tree-vertex-or-root parents that it does not share with any other
reticulation, we also have
\[
  2\, r \;\leq\; t \,,
\]
and it follows that $r \leq n - 1$.

Now, assume that a tree-child network has less that $n - 1$ reticulations.
It then has $t > 2\,r$ vertices with two children. Since a reticulation
is shared by two of these $t$ vertices,
at least one of them has no reticulation child.
\end{proof}

The second ingredient in our proof of Theorem~\ref{thmMinTCN} is the
following property, which is specific to tree-child networks.

\begin{lemma} \label{lemmaInternalReticulations}
Let $N$ be a tree-child network. If $N$ has a
reticulation whose child is not a leaf, then there exists a tree-child network
$N^\star$\! with the same number of leaves as $N$ and such that
$B_2(N^\star) < B_2(N)$.
\end{lemma}

\begin{proof}
Assume that $r$ is a reticulation whose child~$u$ is a tree vertex. Let
$v$ and $w$ be the children of $u$, and $e$ and $e'$ the two incoming edges of
$r$. Finally, let $N'$ and $N''$ be the tree-child networks obtained by removing
$r$ and $u$ from $N$, and:
\begin{itemize}
  \item in the case of $N'$, making~$e$ point to~$v$ and~$e'$ point to~$w$;
  \item in the case of $N''$, making~$e$ point to~$w$ and~$e'$ point to~$v$.
\end{itemize}
This construction is illustrated in
Figure~\ref{figProofLemmaInternalReticulations}.
Note that $N'$ and $N''$ are indeed tree-child networks, because $N$ is
a tree-child network and, therefore, neither the parents nor the siblings
of~$r$ are reticulations.

\begin{figure}[h!]
  \centering
  \captionsetup{width=0.95\linewidth}
  \includegraphics[width=0.75\linewidth]{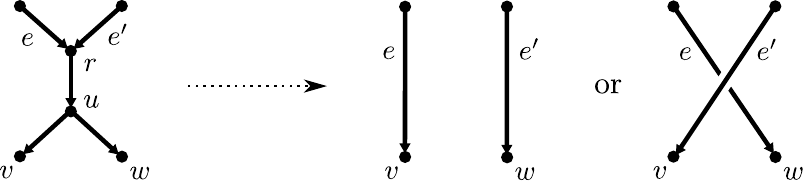}
  \caption{The transformation of $N$ in to $N'$ and $N''$}  
\label{figProofLemmaInternalReticulations}
\end{figure}

Now, let $X = (X_k)$ denote the simple random walk started from the
root of~$N$.
Since $N$ has no directed cycles, $X$ visits $u$ at most one time.
When this happens, it goes 
through $\vec{uv}$ with probability $1/2$ and
through $\vec{uw}$ with probability~$1/2$. Thus, if we let 
$\tilde{X}$ be the random walk induced
by $X$ on $N\setminus\Set{r, u}$; $X'$ and $X''$ be simple random walks
started from the root of $N'$ and $N''$, respectively; and
$Y \sim \mathrm{Bernoulli}(1/2)$ be independent of $(X', X'')$, then
\[
  \tilde{X} \;\sim\; Y X' \;+\; (1 - Y) X'' \,.
\]
In particular, the probabilities $p_\ell$, $p'_\ell$ and $p''_\ell$ of reaching
$\ell$ in $N$, $N'$ and $N''$ satisfy
\[
  p_\ell \;=\; \frac{1}{2}\big( p'_\ell \;+\; p''_\ell\big) \,.
\]
Therefore, by the same concavity argument as in the proof of
Lemma~\ref{lemmaAddEdge},
\[
  B_2(N') \;+\; B_2(N'') \;\leq\; 2\, B_2(N) \,.
\]
Note however that this time the inequality is not strict, because if 
the probability of going through $e$ is the same as the probability of
going through $e'$, then $p'_\ell = p''_\ell$ for every leaf.
In order to get a strict inequality, choose $\tilde{N} \in \Set{N', N''}$
such that $B_2(\tilde{N}) \leq B_2(N)$, and note that $\tilde{N}$ has
strictly less than $n - 1$ reticulations, since it has one reticulation less
than $N$. As a result, by Lemma~\ref{lemmaSaturatedTCN} $\tilde{N}$ has at
least one vertex with two non-reticulation children; and we can
apply Lemma~\ref{lemmaAddEdge} to get a tree-child network $N^\star$ such
that
\[
  B_2(N^\star) \;<\; B_2(\tilde{N}) \;\leq\; B_2(N)\,,
\]
finishing the proof.
\end{proof}

With Lemmas~\ref{lemmaAddEdge} and~\ref{lemmaInternalReticulations}, we
can now prove Theorem~\ref{thmMinTCN} -- i.e.\ show that the fat caterpillar
is the only tree-child network that minimizes $n$.

\begin{proof}[Proof of Theorem~\ref{thmMinTCN}]
Let $N$ be a tree-child network with $n$ leaves that minimizes~$B_2$.
Then, $N$ has $n - 1$ reticulations (otherwise by
Lemma~\ref{lemmaSaturatedTCN} it would have a vertex with two non-reticulated
children and we could use Lemma~\ref{lemmaAddEdge} to contradict
its minimality). Moreover, by Lemma~\ref{lemmaInternalReticulations}
the children of each of these reticulations are all leaves. As a result,
the tree vertices of $N$ are aligned on a single path, as represented in
Figure~\ref{figSinglePath}.A (to see this, start from the root and
follow the edges that point to tree vertices until a leaf is reached; since
no reticulation subtends a tree vertex, the path thus obtained contains all
tree vertices).

\begin{figure}[h!]
  \centering
  \captionsetup{width=0.95\linewidth}
  \includegraphics[width=0.95\linewidth]{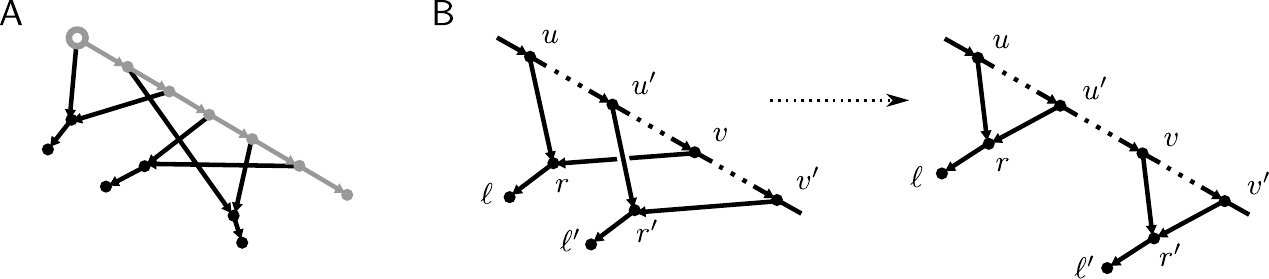}
  \caption{\textsf{A}, an example of a tree-child network that has $n-1$
  reticulations for $n$ leaves and where the child of every reticulation is a
  leaf (such tree-child networks are sometimes known as \emph{one-component
  tree-child networks}~\cite{CardonaZhang2020}). The path containing all tree
  vertices is highlighted. \textsf{B}, the swapping of edges described in the
  proof, for one particular configuration of $u, u', v$ and $v'$ where
  $u$ is an ancestor of $v'$ and $u'$ is an ancestor of $v$.}
\label{figSinglePath}
\end{figure}

Now, assume that $N$ has two reticulations $r$ and $r'$ with parents
$u, v$ and $u', v'$, respectively, such that $u$ is an ancestor of $v'$
and $u'$ is an ancestor of $v$. In that case, let $N'$ be the tree-child
network obtained by replacing the edges
$\vec{vr}$ and $\vec{u\lowprime r\lowprime}$ by $\vec{u \lowprime r}$ and
$\vec{vr\lowprime}$. 
Letting $\ell$ denote the child of $r$ and $\ell'$ that of $r'$, we have:
in $N$, $(p_\ell,\, p_{\ell'}) = (p_u + p_v,\, p_{u'} + p_{v'})$; and, in
$N'$, $(p'_\ell,\, p'_{\ell'}) = (p_u + p_{u'},\, p_v + p_{v'})$. Thus,
by Proposition~\ref{propDeltaB2}, 
$B_2(N') - B_2(N)$ has the same sign as
\[
  (p_v - p_{u'})(p_u - p_{v'}) \;<\; 0 \,.
\]
This shows that if $N$ minimizes $B_2$ then it does not have two reticulations
$r$ and $r'$ such that one parent of $r$ is an ancestor of $r'$ and one
parent of $r'$ is an ancestor of $r$ -- in other words, it is the fat caterpillar
$\FatCat{n}$.

Finally, the $B_2$ index of $\FatCat{n}$ can be computed by noting that
$B_2(\FatCat{2}) = 2 - \frac{3}{4} \log_2(3)$ and that $\FatCat{n+1}$ is
obtained by grafting $\FatCat{2}$ on the leaf with probability $4^{-n+1}$ of
$\FatCat{n}$. Using Proposition~\ref{propGrafting} and a little algebra then
yields the lower bound in Theorem~\ref{thmMinTCN}.
\end{proof}

\vspace{2ex} %HARCODED

\section{Properties of $B_2$ in random trees} \label{secRandomTrees}

In this section, we study the mean and variance of $B_2$ for two of the most
prominent families of random trees: Galton--Watson
trees and Markov branching trees.
\revision{To put the results on Markov branching trees in context,
the reader is referred to Chapter~I of Luc\'{i}a Rotger's PhD dissertation
\cite{LuciaRotger}, which contains a discussion of similar results for other
balance indices}.

\subsection{Galton--Watson trees} \label{secGW}

We consider general Galton--Watson trees whose generations are
indexed by $t \geq 0$, generation~$0$ being the root, 
and where the number of offspring of individual~$i$ from generation~$t$ is
$Y_t(i) \sim Y_t$. The number $Z_t$ of individuals in generation~$t$ is
therefore given by $Z_0 = 1$ and
\[
  Z_{t + 1} \;=\; \sum_{i = 1}^{Z_t} Y_t(i) \,.
\]
The natural filtration of the process is
{$\mathcal{F}_t = \sigma(Y_s(i) : i \geq 1, 0 \leq s \leq t - 1)$}.

%\pagebreak %HARDCODED

\begin{theorem} \label{thmExpecGW}
Let $T$\! be a time-inhomogeneous Galton--Watson tree with offspring distribution
$Y_t$ in generation~$t$, and let
\[
  \alpha_t = \Prob{Y_t > 0}
  \quad\text{and}\quad
  \eta_t = \Expec{\log_2(Y_t) \, \Indic{Y_t > 0}}.
\]
Then, letting \,$\Ball{T}{k}$ denote the restriction of \,$T$\! to 
generations $t = 0, \ldots, k$,
\[
  \Expec{B_2(\Ball{T}{k})} \;=\;
  \sum_{t = 0}^{k - 1} \eta_t \prod_{s = 0}^{t - 1} \alpha_s \,.
\]
In particular, if \,$T$\! is time-homogeneous, so that $\alpha_t = \alpha$ and
$\eta_t = \eta$ for every~$t$,
\[
  \Expec{B_2(\Ball{T}{k})} \;=\;
  \begin{dcases}
  \frac{\eta}{1 - \alpha} \big(1 - \alpha^k\big)
    & \text{if } \alpha < 1 \\[0.5ex]
  \; \eta k
    & \text{if } \alpha = 1. 
  \end{dcases}
\]
\end{theorem}

\begin{remark} \label{remInfiniteGW}
Recall that, by Definition~\ref{defB2Infinite} of $B_2$ for infinite trees,
we have $B_2(T) = \lim_{k} B_2(\Ball{T}{k})$. In time-homogeneous
Galton--Watson trees, there is thus a dichotomy between $\alpha < 1$, where
$\Expec{B_2(T)} = \eta / (1 - \alpha)$ is always finite, including in
the supercritical case (where this implies that 
$\Expec{B_2(T) \given \Abs{T} = +\infty}$ is finite);
and $\alpha = 1$, where we always have $\Expec{B_2(T)} = +\infty$
(except in the degenerate case $Y = 1$ a.s., where $T$ is an infinite path
and $B_2(T) = 0$).
\end{remark}

\begin{proof}
Let $P_t(i)$ be the probability that the random walk goes through
individual~$i$ in generation~$t$ (note that this is a random
variable). A bit of book-keeping shows that
\begin{equation} \label{eqProofThmExpecGW}
  B_2(\Ball{T}{t + 1}) \;=\; B_2(\Ball{T}{t}) \;+\;
  \sum_{i = 1}^{Z_t} P_t(i) \log_2\!\big(Y_t(i)\big) \, \Indic{Y_t(i) > 0} \,.
\end{equation}
Moreover, since $Z_t$ and $(P_t(i) : i \geq 1)$ are $\mathcal{F}_t$-measurable
and $(Y_t(i) : i \geq 1)$ is independent of $\mathcal{F}_t$, with
$Y_t(i) \sim Y_t$ for all $i$, 
\[
  \Expec{\,\sum_{i = 1}^{Z_t} P_t(i) \log_2\!\big(Y_t(i)\big) \,
  \Indic{Y_t(i) > 0} \given \mathcal{F}_t} \;=\;
  \Expec{\log_2(Y_t) \,\Indic{Y_t > 0}}\, \sum_{i = 1}^{Z_t} P_t(i)\,.
\]
As a result, taking expectations in~\eqref{eqProofThmExpecGW} we get
\[
  \Expec{B_2(\Ball{T}{t + 1})} \;=\; \Expec{B_2(\Ball{T}{t})} \;+\;
  \Expec{\log_2(Y_t) \,\Indic{Y_t > 0}}\,
  \Expec{\sum_{i = 1}^{Z_t} P_t(i)}\,.
\]
Now, observe that $\sum_{i = 1}^{Z_t} P_t(i)$ is the probability that
the random walk reaches generation~$t$, conditional on $\mathcal{F}_t$.
Its expected value is therefore the total probability that the
random walk reaches generation~$t$, which is also the probability
that it does not get trapped in a leaf for some generation~$s < t$.
As a result,
\[
  \Expec{\sum_{i = 1}^{Z_t} P_t(i)} \;=\;
  \prod_{s = 0}^{t - 1} \Prob{Y_s > 0} \,.
\]
Writing $\alpha_t = \Prob{Y_t > 0}$ and
$\eta_t = \Expec{\log_2(Y_t) \, \Indic{Y_t > 0}}$, we therefore have
\[
  \Expec{B_2(\Ball{T}{t + 1})} \;=\; \Expec{B_2(\Ball{T}{t})} \;+\;
  \eta_t \, \prod_{s = 0}^{t - 1} \alpha_s \,,
\]
and the theorem follows from the fact that
$\Expec*{\normalsize}{B_2(\Ball{T}{0})} = 0$.
\end{proof}

To close this section, we give an expression for the variance of $B_2$ in
binary Galton--Watson trees. In order to present a simple expression -- and
because this is what we need in the rest of this document -- we focus on the
critical case.

\begin{proposition} \label{propVarBinaryGW}
Let $T$\! be a critical binary Galton--Watson tree, that is, assume that
the offspring distribution is $\Prob{Y = 0} = \Prob{Y = 2} = 1/2$. Then,
\[
  \Var{B_2(\Ball{T}{k})} \;=\;
  \frac{4}{3} - 2^{-k+2} + 4^{-k}\mleft(k + \frac{8}{3}\mright)\,.
\]
As a result, $\Var{B_2(T)} = 4/3$.
\end{proposition}

The proof of this proposition relies on standard calculations similar to that
of Theorem~\ref{thmExpecGW}. It can be found in
Section~\ref{appProofVarBinaryGW} of the Appendix.

\vspace{1ex} %HARDCODED

\subsection{Markov branching trees}

Markov branching trees are a general class of random binary trees that were
introduced by Aldous in \citep{Aldous1996}. They have since become
prominent in phylogenetics, where they are mainly known through the
$\beta$-splitting models, a one-parameter family of models that can
generate a wide variety of random tree shapes. In particular, the $\beta$-splitting
models include the ERM model (which generates trees that have the same shape
as Yule trees), the PDA model (which generates uniform leaf-labeled rooted binary
trees) and the AB model (which generates trees that resemble real-world
phylogenies~\cite{Blum2006}).

A Markov branching tree is described by a family
$\mathbf{q} = (q_n)_{n \geq 2}$ of probability distributions, known as the
\emph{root-split distributions},
such that $q_n$ is symmetric on $\Set{1, \ldots, n - 1}$, that is,
$q_n(k) = q_n(n - k)$. If we do not worry about labels, which are irrelevant
for our purposes, a Markov branching tree $T$ with $n$ leaves can be generated
as follows:
\begin{enumerate}
  \item Sample a random variable $K$ according to $q_n$.
  \item Let the two subtrees of $T$ be independent Markov branching trees with
    $K$ and $n - K$ leaves, respectively.
\end{enumerate}
For a complete introduction to Markov branching trees, see
e.g.~\cite{Lambert2017}.

\vspace{2ex} %HARDCODED

\begin{theorem} \label{thmMarkovBranching}
Let $T_n$ be a Markov branching tree with $n$ leaves and
root-split distributions $\mathbf{q} =(q_n)$.
Let $\mu_n = \Expec{B_2(T_n)}$, $s_n = \Expec{B_2(T_n)^2}$ and
$v_n = \Var{B_2(T_n)}$. Then, letting $K_n \sim q_n$, we have the following
recurrence relations:
\begin{mathlist}
\item $\mu_n = \Expec{\mu_{K_n}} + 1$
\item $s_n = \frac{1}{2}\,\Expec{s_{K_n}} +
  \frac{1}{2}\,\Expec{\mu_{K_n}\, \mu_{n - K_n}}
  + 2\,\Expec{\mu_{K_n}} + 1$
\item $v_n = \frac{1}{4} \Var{\mu_{K_n} + \mu_{n - K_n}} +
  \frac{1}{2}\,\Expec{v_{K_n}}$
\end{mathlist}
with the initial conditions $\mu_1 = s_1 = v_1 = 0$.
\end{theorem}

%\pagebreak %HARDCODED

\begin{proof}
Let $(T'_k)$ and $(T''_k)$ be two independent sequences of Markov branching
trees with root-split distributions $\mathbf{q}$ and such that,
for all~$k\geq 1$, $T'_k$ and $T''_k$ have $k$ leaves, $T'_1$~and $T''_1$ being
the tree that consists only of the root. 
Letting $T = T' \oplus T''$ denote the tree obtained by 
creating a new root and making it
point to the roots of $T'$ and $T''$, we thus have
\[
  T_n \;=\; T'_{K_n} \oplus\, T''_{n - K_n},
\]
where $K_n \sim q_n$ is independent of $(T'_k)$ and $(T''_k)$.
As a result, by Corollary~\ref{corRootSplitRecursion},
\begin{equation} \label{eqProofMarkovBranching}
  B_2(T_n) \;=\; \frac{1}{2}\big(B_2(T'_{K_n}) + B_2(T''_{n - K_n})\big) + 1.
\end{equation}
Taking expectations and using that $K_n \sim n - K_n$ and that 
$K_n$ is independent of $(T'_k)$ and $(T''_k)$, we get
\[
  \Expec{B_2(T_n)} \;=\; \Expec*{\normalsize}{B_2(T'_{K_n})} \;+\; 1,
\]
which, since $\Expec*{\normalsize}{B_2(T'_{K_n}) \given K_n = k} = \mu_k$
for all $k$, is point~(i).

Point (ii) is proved similarly: from Equation~\eqref{eqProofMarkovBranching},
we get
\begin{align*}
  \Expec{B_2(T_n)^2} \;=\;
  &\tfrac{1}{4} \Expec{B_2(T'_{K_n})^2} \;+\;
  \tfrac{1}{4} \Expec{B_2(T''_{n - K_n})^2} \;+\; \\
  &\tfrac{1}{2} \Expec{B_2(T'_{K_n})B_2(T''_{n - K_n})} \;+\;
  \Expec{B_2(T'_{K_n})} \;+\;
  \Expec{B_2(T''_{n - K_n})} \;+\; 1,
\end{align*}
which, by the symmetry of $K_n$ and its independence from $(T'_k)$ and
$(T''_k)$, gives
\[
  \Expec{B_2(T_n)^2} \;=\;
  \tfrac{1}{2} \Expec{B_2(T'_{K_n})^2} \;+\;
  \tfrac{1}{2} \Expec{B_2(T'_{K_n})B_2(T''_{n - K_n})} \;+\;
  2\,\Expec{B_2(T'_{K_n})} \;+\; 1.
\]
Writing these expectations as $\Expec{\Expec{ \;\cdot\, \given K_n}}$
and using the notation from the statement of the theorem, this gives point~(ii).
  
Finally, point~(iii) is obtained by applying the law of total variance to
Equation~\eqref{eqProofMarkovBranching}:
\begin{align*}
  \Var{B_2(T_n)} \;=&\; \tfrac{1}{4}\Var{B_2(T'_{K_n}) + B_2(T''_{n - K_n})}\\[1ex]
  \;=& \; \tfrac{1}{4}\Var{\Expec{B_2(T'_{K_n}) + B_2(T''_{n - K_n})\given K_n}}\\
     &\; +\tfrac{1}{4}\,\Expec{\Var{B_2(T'_{K_n}) + B_2(T''_{n - K_n})\VarGiven K_n}}.
\end{align*}
In this last expression,
\[
  \Expec{B_2(T'_{K_n}) + B_2(T''_{n - K_n})\given K_n} \;=\;
  \mu_{K_n} \;+\; \mu_{n-K_n}\,.
\]
Likewise, since $B_2(T'_{K_n})$ and $B_2(T''_{n - K_n})$ are independent conditional
on $K_n$,
\begin{align*}
  \Var{B_2(T'_{K_n}) + B_2(T''_{n - K_n}) \VarGiven K_n}
   \;&=\;
  \Var{B_2(T'_{K_n})\VarGiven K_n} \;+\;
  \Var{B_2(T''_{n - K_n}) \VarGiven K_n} \\
   \;&=\; v_{K_n} \;+\; v_{n - K_n}, 
\end{align*}
and, taking expectations and using the symmetry
of~$K_n$, we get
\[
  \Expec{\Var{B_2(T'_{K_n}) + B_2(T''_{n - K_n})\VarGiven K_n}}
  \;=\; 2\, \Expec{v_{K_n}} \,.
\]
Putting the pieces together, this yields the required expression for $\Var{B_2(T_n)}$
and finishes the proof.
\end{proof}

\begin{remark}
Equation~\eqref{eqProofMarkovBranching} is known as a
random recursive equation. These can be studied with
the so-called \emph{contraction method}, which often makes it possible to 
characterize (the appropriately rescaled) limit of a random sequence as the
solution of a distributional equation -- see e.g.\
\cite{roesler2001contraction}.
This has been done to study the
limiting distribution of the Sackin and Colless indices under the ERM and PDA
model \cite{Blum2006a}. Doing something similar with $B_2$ might be possible
-- although, as the simulations of Section~\ref{secBiological} show, the
situation will be more complex and there
will be no central limit theorem.
\end{remark}

In the rest of this section, we use the recurrence relations of
Theorem~\ref{thmMarkovBranching} to study the expected value and the variance
of $B_2$ under what are undoubtedly the two most important models of random
trees in mathematical phylogenetics: the ERM\,/\,Yule model and the
PDA\,/\,uniform model.

\begin{definition} \label{defERM}
The ERM model is the Markov branching model whose root-split distributions
are given by
\[
  \forall k\in \Set{1, \ldots, n - 1}, \quad
  q_n(k) \;=\; \frac{1}{n - 1} \, . \qedhere
\]
\end{definition}

What makes the ERM so central to mathematical biology is that it generates
trees that have the same shape as the genealogical tree associated to the
Yule process, i.e.\ the pure-birth process where every individual gives birth
at constant rate~1. This is also the random tree shape associated to the
Kingman coalescent (where every pair of lineages coalesces at rate~1;
\cite{Kingman1982}) and to the genealogy of extant individuals in the Moran
model (where at each step an ordered pair of individuals is sampled uniformly
at random for the first one to be replaced by a copy of the
second~\cite{Moran1958}).  It corresponds to the uniform distribution on the
set of ranked binary trees with $n$ labeled leaves.  This variety of
constructions explains why this model arises in a wide range of biological
applications as well as in many mathematical problems.

\begin{theorem}\label{thmYule}
Let $T_n$ be a tree with $n$ leaves sampled under the ERM\,/\,Yule model.
\begin{mathlist}
\item $\displaystyle \Expec{B_2(T_n)} = \sum_{k = 1}^{n-1} \frac{1}{k}$.
\item $\Var{B_2(T_n)} \,\tendsto{n\to\infty}\; 2 - \pi^2/6$.
\end{mathlist}
\end{theorem}

\begin{proof}
One way to prove point~(i) is to show that $\mu_n = \sum_{k = 1}^{n - 1}1/k$ 
is indeed the solution of the recurrence for the expected value of
$B_2$ given in Theorem~\ref{thmMarkovBranching}.
  
However, a perhaps more intuitive way to obtain $\Expec{B_2(T_n)}$
it is to note that, in the Yule model,
$T_n$ is obtained by grafting a cherry on a leaf $L \in T_{n - 1}$, sampled
uniformly and independently of the shape of $T_{n - 1}$. Let 
$P_\ell$ denote the  probability of reaching a fixed leaf $\ell \in T_{n - 1}$
(note that this is a random variable). Then, by Corollary~\ref{corGraftingCherry}, 
$B_2(T_n) = B_2(T_{n - 1}) + P_L$ and therefore
\[
  \Expec{B_2(T_n)} \;=\; \Expec{B_2(T_{n - 1})} \;+\; \Expec{P_L}\,.
\]
Since $L$ is independent of $(P_\ell)$, $\Expec{P_L} = \Expec{p_L}$,
where $p_\ell = \Expec{P_\ell} = 1 / (n - 1)$, by exchangeability.
As a result, 
\[
  \Expec{B_2(T_n)} \;=\; \Expec{B_2(T_{n - 1})} \;+\; \frac{1}{n - 1}\,, 
\]
and we use that $\Expec{B_2(T_1)} = 0$ to conclude.

Let us now turn to point~(ii). Perhaps surprisingly given the simplicity of
the derivation of the expected value of $B_2(T_n)$, we could not find a way
to obtain the limit of its variance other than solving the recurrence relations of
Theorem~\ref{thmMarkovBranching} explicitly.

Let $H_n = \sum_{k = 1}^n 1/k$ denote the harmonic numbers, and
$H_n^{(m)} = \sum_{k = 1}^n 1/k^m$ the generalized harmonic numbers of
order~$m$. Letting $v_n = \Var{B_2(T_n)}$, 
point~(iii) of Theorem~\ref{thmMarkovBranching} can be written
\begin{equation}\label{eqProofYuleRec}
  v_n \;=\; \alpha_{n - 1} \;+\; \frac{1}{2(n - 1)}\,\sum_{k = 1}^{n - 1} v_k \,, 
\end{equation}
where $\alpha_{n - 1} = \frac{1}{4}\,\Var{H_{K_n - 1} + H_{n - 1 - K_n}}$, with
$K_n \sim \mathrm{Uniform}(\Set{1, \ldots, n - 1})$. Note that this already
shows that if $v_n$ has a finite limit $\ell$, then
$\ell = 2 \lim_n \alpha_n$. Indeed, in that case
$\frac{1}{n - 1}\sum_{k = 1}^{n - 1} v_k \to \ell$, by Cesàro's lemma, and
therefore $\ell$ satisfies $\ell = \lim_n \alpha_n + \ell / 2$.

Let us begin by computing $\alpha_{n - 1}$ explicitly. First,
since $K_n \sim n - K_n$, 
\begin{align} \label{eqProofYuleAlpha}
  &\;\;\Var{H_{K_n - 1} + H_{n - K_n - 1}} \nonumber \\
  \;=&\quad
  2\mleft(\Expec*{\normalsize}{H_{K_n - 1}^2} \,+\,
  \Expec*{\normalsize}{H_{K_n - 1}\,H_{n - 1 - K_n} } \,-\,
  2\,\Expec*{\normalsize}{H_{K_n - 1}}^2 \mright).
\end{align}
Moreover, the following identities for sums of harmonic numbers are well-known;
see e.g.\ Section~1.2.7 of \cite{Knuth1998}:
\begin{itemize}
  \item $\Expec*{\normalsize}{H_{K_n - 1}^2} = H_{n - 1} H_{n - 2}
    + 2 - 2 H_{n - 1}$.
    \hfill \cite[\S 1.2.7, Eq.~(8)]{Knuth1998}
  \item $\Expec{H_{K_n - 1}}^2 = (H_{n - 1} - 1)^2$.
    \hfill \cite[\S 1.2.7, Exercise~15]{Knuth1998}
  \item $\Expec{H_{K_n - 1}H_{n - 1 - K_n}} = 1 - H_{n - 1}^{(2)} +
    (H_{n - 1} - 1)^2$. \hfill \cite[Theorem~1]{Wei2013}
    % Also \cite[\S 1.2.7, Exercise~22]{Knuth1998}
\end{itemize}
Plugging these in~\eqref{eqProofYuleAlpha}, after some simplifications
we get
\begin{equation} \label{eqProofYuleAlphaExplicit}
  \alpha_{n - 1} \;=\; \frac{1}{2} \mleft(2 \;-\; H^{(2)}_{n - 1} \;-\;
  \frac{H_{n - 1}}{n - 1}\mright)\,.
\end{equation}

Let us now solve the recurrence~\eqref{eqProofYuleRec} in order to get
an explicit expression for $v_n$. We start by rearranging the terms in
order to get a first-order recurrence:
\begin{align*}
  v_{n + 1} \;&=\; \alpha_{n} \;+\; \frac{1}{2n}\,\sum_{k = 1}^{n} v_k \\
      \;&=\; \alpha_{n} \;+\; \frac{v_n}{2n} \;+\;
      \frac{n - 1}{n} \cdot\frac{1}{2(n - 1)} \sum_{k = 1}^{n - 1}v_k \\
      \;&=\; \alpha_{n} \;+\; \frac{v_n}{2n} \;+\;
      \frac{n - 1}{n} \big(v_n - \alpha_{n - 1}\big) \,.
\end{align*}
As a result, letting $\beta_n = \alpha_n - \alpha_{n-1} (n - 1) / n$, we see
that $v_n$ is the solution of
\[
  v_{n+1}  \;=\; \frac{2n - 1}{2n}\, v_n \;+\; \beta_n\,, 
\]
with the initial condition $v_1 = 0$. Solving this first-order recurrence
then yields
\begin{equation} \label{eqProofYuleFinal}
  v_n \;=\; \frac{(2n - 1)!!}{(2n)!!}\, \sum_{k = 2}^{n - 1}
  \beta_k \frac{(2k)!!}{(2k - 1)!!}\,,
\end{equation}
where $n!!$ denotes the double factorial.
Finally, to see that $v_n \to c = 2 - \pi^2/6$, note that
$\beta_k \sim c / (2 k)$ and recall that ${(2k)!!}/{(2k - 1)!!} \sim \sqrt{\pi k}$.
The summands in the expression of $v_n$ therefore
are asymptotically equivalent to $(c \sqrt{\pi}) / 2\sqrt{k}) $, and the
result follows from a standard application of the integral test for convergence,
since $\int\! \frac{1}{\sqrt{x}}\,dx = 2 \sqrt{x}$.
\end{proof}

\begin{definition} \label{defPDA}
The PDA model is the Markov branching model whose root-split distributions
are given by
\[
  \forall k\in \Set{1, \ldots, n - 1}, \quad
  q_n(k) \;=\; \frac{1}{2} \binom{n}{k} \frac{t_k\, t_{n - k}}{t_n} \,,
\]
where $t_n = (2n - 3)!!$\, is the number of 
rooted binary trees with $n$ labeled leaves.
\end{definition}

What makes the PDA model stand out is that it generates trees that are uniformly
distributed on the set of rooted binary trees with $n$ labeled leaves. For this
reason, it is sometimes referred to as the ``uniform model'' and, in
phylogenetics, it is the standard alternative to the Yule model when in need of
a null model.

\begin{theorem}\label{thmUniform}
Let $T_n$\! be a tree with $n$ leaves sampled under the PDA model or,
equivalently, uniformly on the set of rooted binary trees with $n$ labeled leaves.
\begin{mathlist}
\item $\Expec{B_2(T_n)} = 3\, (n - 1) / (n + 1)$.
\item $\Var{B_2(T_n)} \,\tendsto{n\to\infty}\;4/9$.
\end{mathlist}
\end{theorem}

\begin{proof}
By Theorem~\ref{thmMarkovBranching}, to prove~(i) it suffices to show
that $\mu_n = 3\, (n - 1) / (n + 1)$ is indeed the solution of
$\mu_1 = 0$ and
\begin{equation} \label{eqProofPDA}
  \mu_n \;=\; 1 + \sum_{k = 1}^{n - 1} \mu_k \,q_n(k)\,, 
\end{equation}
where
\[
  q_n(k) \;=\; \frac{1}{2} \binom{n}{k}
  \frac{(2k - 3)!!\,(2(n - k) - 3)!!}{(2n - 3)!!} \,.
\]
Assume that $\mu_k = 3\, (k - 1) / (k + 1)$ for all $k \in \Set{1, \ldots n - 1}$.
Note that \revision{this can be written as} $1 + \mu_k = 2(2k - 1)/(k+1)$, so that
\[
  \big(1 \;+\; \mu_k\big)\, q_n(k) \;=\;
  2\;\frac{2n - 1}{n + 1}\, q_{n + 1}(k + 1)\, .
\]
Plugging this in Equation~\eqref{eqProofPDA}, we get
\begin{align*}
  \mu_n \;=\; \sum_{k = 1}^{n - 1} \big(1 + \mu_k\big)\,q_n(k)
   &=\; 2\,\frac{2n - 1}{n + 1}\,\sum_{k = 1}^{n - 1} q_{n + 1}(k + 1)\, \\
   &=\; 2\,\frac{2n - 1}{n + 1}\,\big(1 - q_{n + 1}(1)\big)
\end{align*}
which, since $q_{n + 1}(1) = \frac{n + 1}{2(2n - 1)}$, yields
$\mu_n = 3(n - 1) / (n + 1)$.

To prove point~(ii), recall that, as $n \to \infty$, the uniform rooted binary
tree with $n$ labeled leaves converges in distribution to the size-biased
Galton--Watson tree~$\hat{T}$ obtained by grafting independent
critical binary Galton--Watson trees on each leaf of the infinite caterpillar,
as illustrated in Figure~\ref{figSizeBiased}.
See~\cite{Janson2012} for a complete introduction to the subject.

\begin{figure}[h!]
  \centering
  \captionsetup{width=0.85\linewidth}
  \includegraphics[width=0.45\linewidth]{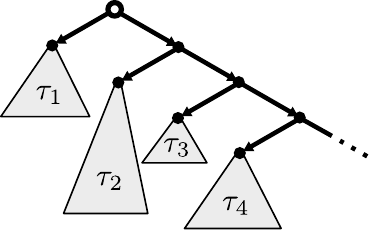}
  \caption{Construction of the size-biased Galton--Watson tree (for
  the critical binary Galton--Watson tree).
  Let the root be the end of a
  one-way infinite path known as the \emph{spine} of the tree,
  and let every vertex from the spine point to the root of an independent
  Galton--Watson tree with $2\,\mathrm{Bernoulli}(2)$ offspring distribution.} 
\label{figSizeBiased}
\end{figure}

Letting $\tau_k$ denote the Galton--Watson tree grafted on the leaf at depth~$k$
of the infinite caterpillar, by Proposition~\ref{propGrafting} we have
\[
  B_2(\hat{T}) \;=\; \sum_{k \geq 1} 2^{-k} B_2(\tau_k) \,.
\]
Because the variables $B_2(\tau_k)$ are independent,
this gives
\[
  \Var[\normalsize]{B_2(\hat{T})} \;=\; \sum_{k\geq 1} 4^{-k} \Var{B_2(\tau_k)}
\]
and, since we have seen in Proposition~\ref{propVarBinaryGW} that
$\Var{B_2(\tau_k)} = 4/3$,
\[
  \Var[\normalsize]{B_2(\hat{T})} = \frac{4}{9}\,.
\]
Since $\Var{B_2(T_n)} \to \Var[\normalsize]{B_2(\hat{T})}$, this concludes the
proof.
\end{proof}

\begin{remark} \label{remAltProofGW} 
Recalling that $T_n$ has the same distribution as the Galton--Watson
tree with $2\,\mathrm{Bernoulli}(p)$ offspring distribution conditioned
on having $n$ leaves, and that the probability that the 
$2\,\mathrm{Bernoulli}(p)$-Galton--Watson tree has $n$ leaves is
\[
  \varpi_n \;=\; 2^{n-1} (2n - 3)!!\, p^{n - 1} (1 - p)^{n} / n! \, , 
\]
point~(i) from Theorem~\ref{thmUniform} can be used to give an alternative
derivation of the expected value of $B_2$ in binary Galton--Watson trees
(which we already have as a special case of Theorem~\ref{thmExpecGW}).
Indeed, it is readily checked that
\[
  \sum_{n \geq 1} 3\, \frac{n - 1}{n + 1} \varpi_n \;=\; \frac{p}{1 - p} \,. 
  \qedhere
\]
\end{remark}

\vspace{1ex} %HARDCODED

\section{Biological relevance: empirical study} \label{secStats}
 
\subsection{Design of the study}

We follow the approach originally developed by Kirkpatrick and Slatkin
(henceforth K\&S)
in~\cite{Kirkpatrick1993} and later extended
by Agapow and Purvis (A\&P) in~\cite{Agapow2002}.
The idea is to compare how good are various balance indices
at distinguishing between trees of different origins.
For this, we pick a null model to generate random trees and
use it to build empirical confidence intervals for each balance index.  We then
consider sets of trees that were not generated by this null model and, for
each balance index, compute the percentage of those trees whose balance
index  does not fall in the confidence interval obtained under the null model.
This percentage is used as a direct measure of the ``power'' of the balance
index: the higher it is, the more efficient the balance index was at
distinguishing the trees of the test set from those generated by the null
model.

Although the general strategies are the same, our analysis differs
from those of K\&S and A\&P in three important respects:
\begin{enumerate}
  \item We use several, varied null models.
  \item We use real-world phylogenies for our test sets.
  \item We consider a wider range of tree sizes.
\end{enumerate}

\paragraph{Null models:}

Both K\&S and A\&P use the ERM model as their null model. This is standard,
but debatable: indeed, this
model is well-known to produce trees that are far more balanced than real-world
phylogenies~\cite{Blum2006} (in fact, A\&P acknowledge this and
consider one-sided confidence intervals as a result).  Thus,
both studies focus on rejecting a single,
relatively unrealistic null hypothesis. 

To address this issue, we consider 5 different null models which,
together, cover a large variety of tree shapes: the ERM model; the PDA model;
the Aldous branching model (AB); and two versions of a
multiplicative fitness landscape
birth process (MFL). In this last model, each lineage has it own, evolving
birth rate. The birth rates do not spontaneously change over time, but when
a birth occurs the birth rate of each of the two newly formed
lineages is inherited from their mother and then multiplied by $(1+s)$ with
probability $p$, independently of each other and of everything else. 
We consider two variants MFL1 and MFL2 of this model, where the parameters are 
$s = 0.25$ and $p = 0.5$ for MFL1, and $s = 0.5$ and $p = 0.5$ for MFL2.
The choice of these parameters is somewhat arbitrary:
our goal is mainly to try less conventional null models. We did check that the
trees generated by these models seem biologically relevant;
in fact for the range of the number of leaves considered, they produce trees
that are either slightly less (MFL1) or slightly more (MFL2) balanced than the
trees produced by the AB model. However, in order not to risk biasing our
study, we did not try to optimize the parameters $s$ and $p$ based on a specific
balance index.

For each null model and each value of the number of leaves,
two-sided (symmetric) 95\% confidence intervals were built empirically using
$10^5$ replicates. The code for doing so and its output can be downloaded
from~\cite{ZENODO}.

\paragraph{Test sets:}

The test sets used by A\&P consist of trees generated by various ad hoc
models of random trees. This is ideal to isolate the effect of a specific
biological mechanism
on the performance of the balance indices. However, this also
leaves the biological relevance of the trees open to debate.

Since we are more interested in assessing the potential relevance of
various balance indices in real-world applications than in precisely
characterizing the situations in which each of them should be used,
we use real phylogenetic trees to constitute our test sets.
Conceptually, these can be seen as samples from ``black box''
models of random trees. Thus, instead of
having models whose inner workings we understand but whose relevance
is questionable, we have  models that we know nothing about
but whose relevance is undeniable.

In total, we used trees from 4 databases: 4378 trees from
TreeBASE~\cite{TreeBASE}; 14509 trees from OrthoMaM~\cite{OrthoMaM}; 77843
trees from PhylomeDB~\cite{PhylomeDB}; and 85581 trees from
HOGENOM~\cite{HOGENOM}. More details on how these trees were obtained
can be found in~\cite{ZENODO}, where the complete list of trees that we used
can also be downloaded in Newick format.

\paragraph{Range of the number of leaves:}

K\&S consider trees with up to 50 leaves, and A\&P with up to 64 leaves. While
these were considered to be large trees at the time, this is not
so much the case today. In fact, 12\% of the phylogenetic trees that we
collected have between 50 and 100 leaves; and 20\% have more than
100 leaves. Moreover, larger trees are bound to become more and more common as
sequencing data and reconstitution methods improve.

In order to see how the size of the trees affects the performance of each
balance index, we grouped the trees of each dataset in categories based on
their number of leaves: very small trees {($6\leq n \leq 12$)}; small trees
{($12 \leq n \leq 24$)}; intermediate trees {($25\leq n \leq 50$)}; large trees
($50 \leq n \leq 100$); and very large trees {($n > 100$)}. 
Note that these size categories are not related to the estimation of the
confidence intervals (which is done for every possible value of $n$):
the point of these categories is to have enough trees
of similar sizes to compute a ``by-size'' average of the proportion of trees
rejected by each balance index.  Also note that the slight overlap between some
of the categories is not a problem (these categories can be thought of as
playing the role of windows in a moving average).

In the case of the OrthoMaM database, only the \emph{large} and \emph{very
large} categories were considered because the other categories did not contain
enough trees to get a reliable estimate of the proportion of trees rejected by
each balance index. All categories that were included in our
analysis contain at least 500 trees -- usually several thousands. Our
estimation of the size-specific proportions of trees rejected by the
balance indices should therefore be reasonably reliable. 

\paragraph{Balance indices:}

The set of balance indices that we compare is very similar to that used by
K\&S and by A\&P, as well as other studies that compare balance
indices \cite{Maia2004,Matsen2006,Hayati2019}. We consider: the Colless index;
the Sackin index; the number of cherries, as suggested in~\cite{McKenzie2000};
$\sigma_N^2$, the variance of the depths of the leaves;
the $B_1$ index; and the $B_2$ index. Let us briefly recall the
definition of the balance indices that we have not yet mentioned in this paper.

The variance of the depths of the leaves was originally suggested as a
measure of phylogenetic balance by Sackin in~\cite{Sackin1972}, but formally
introduced and first used by K\&S in~\cite{Kirkpatrick1993}. It is defined as
\begin{equation} \label{eqSigmaN}
  \sigma_N^2(T) \;=\; \frac{1}{n} \sum_{\ell \in L}(\delta_\ell - \bar{\delta}(T))^2,  
\end{equation}
where the sum runs over the leaves of the tree; $n$ is the number of leaves;
$\delta_\ell$ the depth of leaf~$\ell$; 
and $\bar{\delta}(T) = \frac{1}{n}\sum_\ell \delta_\ell$ the average
leaf depth. Note that $\bar{\delta}(T)$ is a rescaled version of the
Sackin index, and that various rescaled versions are
frequently used instead of
\revision{$\mathrm{Sackin}(T) = \sum_\ell \delta_\ell$}.
Since we estimate our confidence intervals for each value of $n$, such scalings
by a function of $n$ are irrelevant.

The $B_1$ index was introduced by Shao and Sokal in~\cite{Shao1990} and is
defined as
\begin{equation} \label{eqB1}
  B_1(T) \;=\; \sum_{i \in I^*}
  \big(\max\Set{\delta_\ell \suchthat \ell \text{ is subtended by } i}\big)^{-1}, 
\end{equation}
where the sum runs over internal vertices, excluding the root.

\paragraph{Summary:}

Altogether, this gives us 85
(null model; test set; size category)
scenarios under which to compare 6 balance indices.
The complete list of null models, test sets, size categories and
balance indices that we use in our analysis is summarized in
Table~\ref{tabSummaryDesign}.
For each null model, we have the confidence interval of each balance index;
and for each test set and each size category, we have
at least 500 trees. We estimate
the proportion of trees rejected by each balance index in each specific
scenario and use it as a direct measure of the ``statistical power'' of
that index in that scenario. The results are presented in
Section~\ref{secResults}.

\vspace{2ex} %HARDCODED

\begin{table}[h!]
\centering
\captionsetup{width=0.87\linewidth}
\begin{tabular}{llll}
\toprule
Null models    & Test sets                 & Range of $n$   & Balance indices     \\
\midrule
ERM            & TreeBASE \cite{TreeBASE}  & $n\in$  6--12  & Sackin index \cite{Shao1990} \\
(Yule model)   & 4378 trees                &                & see Eq.~\eqref{eqSackin} here\\
               &                           & $n\in$ 12--24  &                     \\
PDA            & OrthoMaM \cite{OrthoMaM}  &                & Colless index \cite{Colless1982} \\
(uniform model)& 14509 trees               & $n\in$  25--50 & see Eq.~\eqref{eqColless}\\
               &                           &                &                     \\
AB             & PhylomeDB \cite{PhylomeDB}& $n\in$  50--100& \# of cherries \cite{McKenzie2000} \\
($\beta$-splitting & 77843 trees           &                &                     \\
with $\beta = -1$) &                       & $n > 100$     & $\sigma_N^2$ index \cite{Sackin1972,Kirkpatrick1993} \\
               &  HOGENOM \cite{HOGENOM}   &               & see Eq. \eqref{eqB1}     \\
MFL1           &  85581 trees              &               &                      \\
(see main text)&                           &               & $B_1$ index \cite{Shao1990} \\
$p=0.5, s=0.25$&                           &               & see Eq. \eqref{eqB1}   \\
               &                           &               &                      \\
MFL2           &                           &               & $B_2$ index \cite{Shao1990} \\
$p=0.5, s=0.5$ &                           &               & see Def.~\ref{defB2}  \\
\bottomrule
\end{tabular}
\caption{Summary of the null models, test sets, size categories and
balance indices used in this study. Note that the columns are independent (that
is, there is no correspondence between the lines of different columns).}
\label{tabSummaryDesign}
\end{table}

In addition to this, we also test whether some pairs of statistics
work better than others. The motivation for this is that, although
biologists often restrict themselves to the Colless and Sackin indices,
it is well documented that these indices are extremely correlated -- both under
theoretical models and in practice
(see e.g.~\cite{Blum2005} as well as Figure~\ref{figConfidenceRegions} below).
As a result, to some extent they contain the same information and it is unclear
whether using them jointly provides a real advantage over using only one of them
with another balance index such as $B_2$.

To assess this, we estimate the statistical power of each
pair of balance indices. The methodology for doing this is exactly the same as
when testing the power of a single balance index -- except that
this time we need 2D confidence regions to decide whether to keep or to reject 
each tree. Because there is no standard way to build such 2D confidence regions
for non-gaussian data (in particular when the joint distribution of
the data is not symmetric, as is the case here; see
Figure~\ref{figConfidenceRegions}),
we had to choose one such method somewhat arbitrarily.
We chose to decompose the observations into convex layers and
pick the largest layer that contains less than 95\% of the observations.
Compared to fitting a density and using its level sets, this
method has the advantage of being less computationally intensive and more
robust when used on data that take discrete values (when fitting a density,
e.g.\ with a kernel density estimation, the smoothing can unpredictably impact
the shape of the resulting level sets).

The convex layers of a set of points
are the nested convex polygons obtained through the following procedure: let
$S_0$ be the complete set of points and $H_0$ be the vertices of its
convex hull. Let then $S_{i + 1} = S_{i} \setminus H_i$ and $H_{i + 1}$ be
the vertices of the convex hull of $S_{i + 1}$, and iterate until
$S_{i+1}$ is empty. This construction is illustrated in
Figure~\ref{figConvexLayers}. See e.g.~\cite{Chazelle1985} for more on 2D
convex layers and how to compute them.

\enlargethispage{3ex} %HARDCODED
Note that because a null model can generate the same tree several times,
and also because different trees can have the same balance indices,
the points that we consider are not necessarily distinct; and this relevant
information should be taken into account.
We thus treat the set of sampled points
as a multiset. When several points are superposed (and therefore
correspond to the same vertex of a convex hull $H_i$), we only remove one of
them from~$S_i$. Our convex layers
can therefore overlap at their vertices, and
each point that is present $k$ times in the data is the vertex of $k$
convex layers.

\pagebreak %HARDCODED

\begin{figure}[h!]
  \centering
  \captionsetup{width=0.95\linewidth}
  \includegraphics[width=0.95\linewidth]{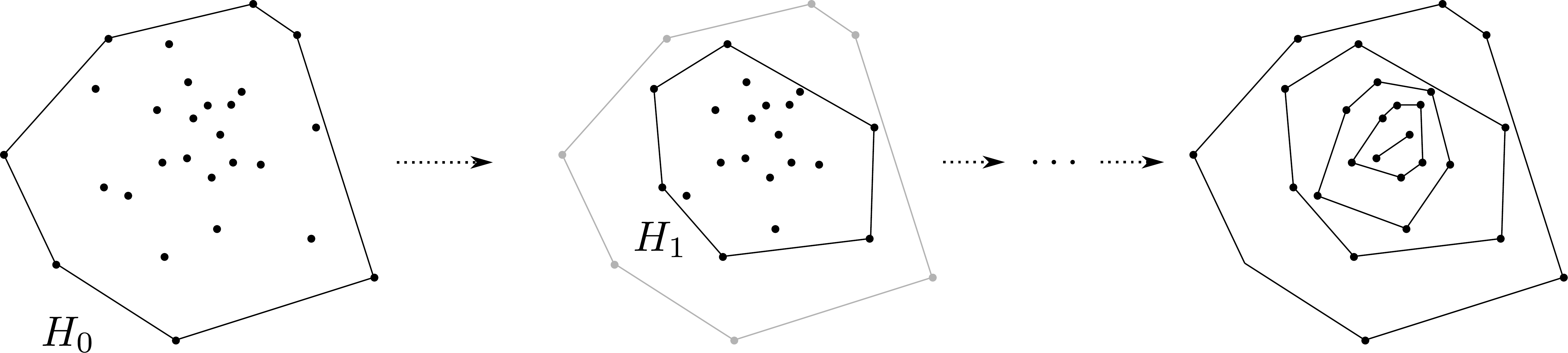}
  \caption{Example of construction of the convex layers decomposition of
  a set of points (here the innermost convex layer is a degenerate polygon with
  two vertices).} 
\label{figConvexLayers}
\end{figure}

Finally, since we stop the construction as soon as a convex layer contains
less than 95\% of the points, our confidence regions do not contain
exactly 95\% of the points generated by the null model. However, because of the
large ($=10^5$) number of points that we use, they always contain between~95\%
and 94.5\% of the points.  Examples of the joint distributions of some balance
statistics and of the corresponding 95\% confidence regions are given in
Figure~\ref{figConfidenceRegions}.
As before, the code used to compute the confidence regions and its output can 
be found in~\cite{ZENODO}.

\begin{figure}[h!]
  \centering
  \captionsetup{width=\linewidth}
  \includegraphics[width=\linewidth]{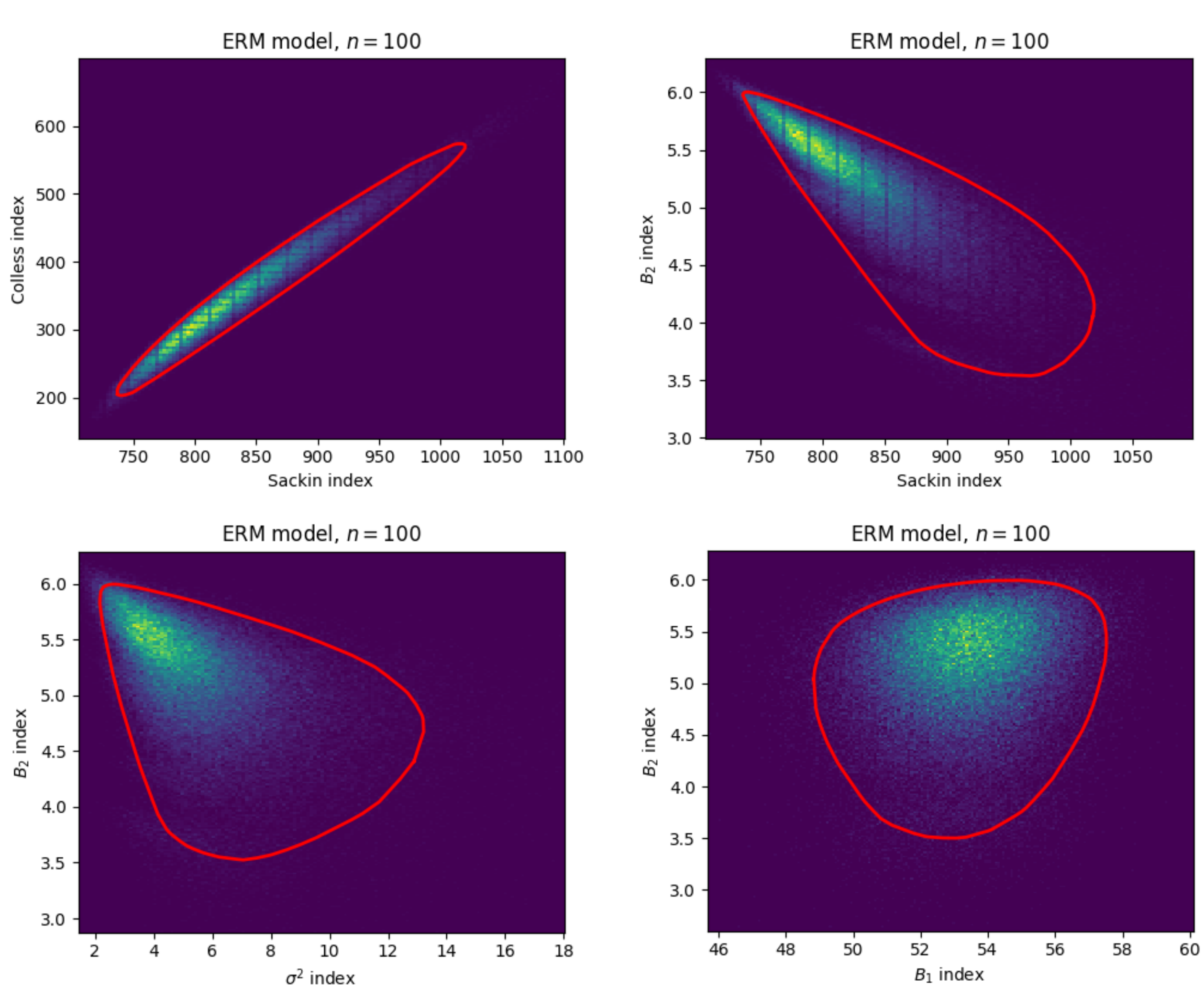}
  \caption{Examples of 2D histograms depicting the joint distributions
  of several balance indices, here under the ERM model for $n = 100$ leaves.
  The color scale is not given because it is not relevant.
  The 95\% confidence regions based on the convex layers
  are represented in red. These histograms and confidence regions were
  build using the same $10^5$ independent realizations of the ERM model.
  Note that the vertical (resp.\ horizontal) stripes visible in the distribution
  of the Sackin (resp.\ Colless) index are artifacts that result from the
  binning used to plot the histograms. This occurs only for the Sackin and Colless
  indices because they take fewer distinct values than the $\sigma_N^2$, $B_1$
  or $B_2$ indices.} 
\label{figConfidenceRegions}
\end{figure}

%\pagebreak %HARDCODED

\subsection{Results} \label{secResults}

Table~\ref{tabRawPower} is an excerpt from the full output of our analysis.
The complete table of results (where the performance of each balance index and
each pair of balance indices is given for each of the 85 scenarios that we
considered) can be downloaded from~\cite{ZENODO}. 
Here we have selected a handful of relevant scenarios to illustrate
specific points; a more systematic analysis will follow.

\begin{table}[h!]
\centering
\captionsetup{width=\linewidth}
\centerline{\small
\begin{tabular}{llcrrrrrrr}
\toprule
 Model & Test set & Range $n$ & Trees&\!\! Sackin &\!\!\!\! Colless &\!\!\!Cherries & $\sigma_N^2$ & $B_1$ & $B_2$ \\
\midrule
  ERM   & HOGENOM    &  6--12   & 39763 & \NS{4.51} & \NS{4.77} & \NS{1.59} &\HIGH 5.59 & \NS{3.40} & \NS{4.60} \\
  ERM   & HOGENOM    &  12--24  & 24493 & 17.64 & 17.91   &  5.35  & \HIGH 18.09 & 16.05 & 13.71 \\
  ERM   & HOGENOM    &  25--50  & 14371 &\HIGH  35.57 & 35.11   &  11.33 & 34.37 & 29.27 & 18.55 \\
  ERM   & HOGENOM    &  50--100 & 7020  &\HIGH  53.13 & 52.45   &  20.98 & 50.84 & 46.85 & 23.80 \\
  ERM   & HOGENOM    &  >100    & 3305  &\HIGH  73.31 & 72.28   &  33.77 & 71.07 & 65.93 & 35.46 \\
  ERM   & PhylomeDB  &  6--12   & 11632 & 31.04 & 30.75   &  9.13  & 26.68 & 19.11 &\HIGH  32.25 \\
  ERM   & PhylomeDB  &  12--24  & 20022 & 71.97 & 69.87   &  12.87 & 58.95 & 41.93 &\HIGH  77.07 \\
  ERM   & PhylomeDB  &  25--50  & 15523 & 94.00    & 92.66   &  21.81 & 82.30  & 59.06 &\HIGH  97.04 \\
  ERM   & PhylomeDB  &  50--100 & 11302 & 98.76 & 98.37   &  29.07 & 90.88 & 66.71 &\HIGH  99.85 \\
  ERM   & PhylomeDB  &  >100    & 21643 & 99.81 & 99.72   &  39.02 & 95.01 & 70.59 &\HIGH  100 \\
  PDA   & OrthoMaM   & 50-100 & 2700 & 35.48 & 41.07 & 18.52 & 66.00 &\HIGH  67.19 & \NS{4.26} \\
  PDA   & OrthoMaM   & >100 & 11741 & 51.40 & 63.09 & 32.92 & 84.43 &\HIGH  91.76 &\NS{ 2.84} \\
  AB    & OrthoMaM   & 50-100 & 2700 & 8.26 & 7.44 & \NS{4.04} & \NS{1.22} & \NS{4.85} &\HIGH  13.85 \\
  AB    & OrthoMaM   & >100 & 11741 & \NS{4.64} & \NS{3.88} & \NS{1.88} & \NS{0.49} & \NS{1.57} &\HIGH  10.33 \\
  MFL1  & OrthoMaM   & 50-100 &2700 & \NS{2.89} &\NS{2.44}  & 20.70 &\NS{0.30} & \HIGH 23.30 &16.96 \\
  MFL1  & OrthoMaM   & >100  & 11741 &\NS{0.68} &\NS{0.56}  & \HIGH 15.19 &\NS{0.06} &12.61& 11.50 \\
  MFL2  & OrthoMaM   & 50-100 & 2700 & \NS{3.15} & 5.70 & 6.89 &\HIGH  16.00 & \NS{1.63} & \NS{4.48} \\
  MFL2  & OrthoMaM   & >100 & 11741 & \NS{4.34} & 7.69 & \NS{3.58} &\HIGH  38.11 & \NS{0.54} & \NS{2.5} \\
\bottomrule
\end{tabular}
}
\caption{Estimated statistical power of each balance index, as measured by the
  percentage of trees of the test set rejected against the null model. The
  scenarios given in this table were hand-picked to illustrate specific points.
 The highlighted values correspond to the best index for a given scenario
 and the greyed out ones to indices with a power $\leq$ 5\%.}
\label{tabRawPower}
\end{table}

First, one can ask whether some balance indices are
consistently better when working with a specific null model. In particular,
in light of the study of Agapow and Purvis we might expect the Sackin and
Colless indices to be markedly more powerful than other balance indices
when testing against the ERM model.
While this general trend is somewhat confirmed by our study,
this is by no means a golden rule: as shown in Table~\ref{tabRawPower},
$B_2$ is better at distinguishing the trees of
PhylomeDB from ERM trees than any other balance index.

%\pagebreak %HARDCODED

Second, one can wonder whether some balance indices are consistently better
with some test sets. Here, comparing the performance of the 
$\sigma_N^2$, $B_1$ and $B_2$ indices on OrthoMaM is
informative: indeed, these indices perform
completely differently depending on the null model used. For instance,
$B_2$ is the only index that fails to distinguish the trees of
OrthoMaM from the PDA model; but it is also the only index that is somewhat able
to distinguish those same trees from the AB model. Similarly, $\sigma_N^2$ is
the only index able to distinguish OrthoMaM from MFL2, and $B_1$ is
the best index to distinguish OrthoMaM from PDA or MFL1. 
The Sackin and the Colless index stand out by their poor
performance when testing OrthoMaM against the AB, MFL1 or MFL2 models.

Table~\ref{tabRawPower} shows that there are no absolute rules when it comes to
ranking the performance of the various balance indices, and looking at the
full output of our analysis only reinforces this conclusion. However, general
trends might still emerge when looking across the 85 scenarios that we
considered.

In order to assess this, we rank the indices from best (1) to worst (6) for
each scenario.  We then we average these ranks over all scenarios that share a
common factor: over all scenarios where the null model is the ERM model; over
all scenarios where the test set is TreeBASE; etc.  These average ranks are
given in Table~\ref{tabResultsUnivariate}.  We use average ranks rather than
average powers because when averaging the powers an excellent performance in
a single scenario might compensate poor performances in several other
scenarios. Using ranks ensures that each scenario gets the same importance.

\begin{table}[h!]
\centering
\captionsetup{width=\linewidth}
\begin{tabular}{lccccccc}
\toprule
Factor & Scenarios & Sackin & Colless & Cherries &
                                  ~~$\sigma^2_N$~~~ & ~~$B_1$~~~ & ~~$B_2$~~~ \\
\midrule
      ERM  &    17  & \HIGH 1.71  &  2.29  &  6.00  &  3.59  &  4.65  &  2.76 \\ 
      PDA  &    17  &  4.47  &  3.38  &  4.18  &  2.74  & \HIGH 1.35  &  4.88 \\
       AB  &    17  &  2.76  &  2.88  &  5.47  &  4.00  &  3.53  & \HIGH 2.35 \\
     MFL1  &    17  &  3.41  &  3.82  &  4.29  &  4.47  &  2.65  & \HIGH 2.35 \\
     MFL2  &    17  &  3.47  &  3.00  &  4.97  &  2.88  &  3.88  & \HIGH 2.79 \\
\midrule
 TreeBASE  &    25  &  3.34  &  3.44  &  4.58  &  3.68  & \HIGH 2.80  &  3.16 \\
 OrthoMaM  &    10  & \HIGH 3.00  &  3.10  &  4.00  &  4.00  &  3.40  &  3.50 \\
PhylomeDB  &    25  &  2.94  &  3.46  &  5.12  &  4.14  &  3.68  & \HIGH 1.66 \\
  HOGENOM  &    25  &  3.28  & \HIGH 2.32  &  5.64  &  2.60  &  3.08  &  4.08 \\
\midrule
$n\in$  6--12  &  15  &  3.17  &  2.77  &  6.00  & \HIGH 2.17  &  4.33  &  2.57 \\
$n\in$ 12--24  &  15  &  3.27  &  2.87  &  5.73  &  3.33  &  3.60  & \HIGH 2.20 \\
$n\in$ 25--50  &  15  & \HIGH 2.93  &  3.07  &  5.07  &  4.00  & \HIGH 2.93  &  3.00 \\
$n\in$ 50--100 &  20  &  3.22  &  3.35  &  4.47  &  4.00  & \HIGH 2.65  &  3.30 \\
$n > 100$      &  20  &  3.20  &  3.20  &  4.10  &  3.90  & \HIGH 2.85  &  3.75 \\
\midrule
      All  &    85  &  3.16  &  3.08  &  4.98  &  3.54  &  3.21  & \HIGH 3.03 \\
\bottomrule
\end{tabular}
  \caption{Average ranks of the balance indices. For each scenario, the
  indices are ranked from 1 (the best) to 6 (the worse). The ranks are then
  averaged over all scenarios that share the factor indicated in the
  left-most column.}
\label{tabResultsUnivariate}
\end{table}

Although a few general trends can be identified (such as the fact
that the Sackin index seems to be consistently better when testing
against the ERM model; or that $B_1$ seems to be the best index when 
testing against the PDA model or working with large trees), the main
takeaway from Table~\ref{tabResultsUnivariate} is that
there does not seem to be a ``silver bullet'' index that would work well
in every situation: with the exception of the number of cherries,
each index works well in some situations and
poorly in some others. This is made apparent by the fact that, when the
ranks are averaged over all 85 scenarios,
they all come out roughly equal to~3.5 --
which is what we expect in the absence of any difference of performance
between the indices.

All things considered, the only reasonably certain conclusions that emerge from
our analysis are the following:
\begin{itemize}
  \item The number of cherries is consistently
    worse than the other balances indices.
  \item Neither the Colless nor the Sackin index are consistently 
    better than other balance indices.
  \item The $B_2$ index is not consistently worse than other
    indices. In fact, its overall performance seems comparable
    to that of the Colless and of the Sackin index.
\end{itemize}

Let us now turn to the results of our bivariate analysis. Recall that,
since the Colless and the Sackin index are currently the standard measures
of phylogenetic balance and are frequently used together,
our specific goal is to test whether other balance indices might
complement the Colless / Sackin indices better than they complement each
other. Table~\ref{tabResultsBivariateColless}
(resp.~\ref{tabResultsBivariateSackin}) gives the average rank of each of the
pairs of indices that include the
Colless (resp.\ Sackin) index, using the same methodology as previously
(note however that this time the ranks go from 1 to 5, so the expected rank
when all indices have the same overall performance is~3).

\begin{table}[h!]
\centering
\captionsetup{width=0.9\linewidth}
\begin{tabular}{lcccccc}
\toprule
  Factor  & \!Scenarios\! & \!(C,\,Sackin)\! & \!(C,\,cherries)\!
  &\,(C,\,$\sigma_N^2$)\, &\,(C,\,$B_1$)\, &\,(C,\,$B_2$)\, \\
\midrule
      ERM  &      17  &   3.79  &   3.24  & \HIGH 2.41  &    2.97  &   2.59 \\  
      PDA  &      17  &   4.41  &   2.91  &  2.74  & \HIGH   1.44  &   3.50 \\
       AB  &      17  &   3.91  &   3.00  & \HIGH 2.00  &    2.85  &   3.24 \\
     MFL1  &      17  &   4.35  &   3.50  & \HIGH 2.00  &    2.74  &   2.41 \\
     MFL2  &      17  &   3.74  &   3.18  & \HIGH 2.06  &    3.11  &   2.91 \\
\midrule
 TreeBASE\!\!\!\!  &      25  &   4.40  &   3.56  & \HIGH 1.94  &    2.54  &   2.56 \\
 OrthoMaM\!\!\!\!  &      10  &   2.65  &   4.00  & \HIGH 2.20  &    3.45  &   2.70 \\
PhylomeDB\!\!\!\!  &      25  &   4.28  &   3.36  & \HIGH 1.52  &    3.56  &   2.28 \\
  HOGENOM\!\!\!\!  &      25  &   4.00  &   2.24  &  3.28  & \HIGH   1.44  &   4.04 \\
\midrule
$n\in$  6--12  &  15  &   4.40  &   3.07  & \HIGH 1.07  &    3.00  &   3.47 \\
$n\in$ 12--24  &  15  &   4.00  &   3.80  &  2.23  &    2.83  & \HIGH  2.13 \\
$n\in$ 25--50  &  15  &   4.20  &   2.80  &  2.73  & \HIGH   2.40  &   2.87 \\
$n\in$ 50--100\!\!\!\! &  20  &   3.92  &   3.15  & \HIGH 2.25  &    2.52  &   3.15 \\
$n > 100$      &  20  &   3.80  &   3.05  &  2.75  &  \HIGH  2.45  &   2.95 \\
\midrule
   All         &  85  &   4.04  &   3.16  & \HIGH 2.24  &    2.62  &   2.93 \\
\bottomrule
\end{tabular}
\caption{Average ranks of the pairs of balance indices that include the
  Colless index.}
\label{tabResultsBivariateColless}
\end{table}

\begin{table}[h!]
\centering
\captionsetup{width=0.9\linewidth}
\begin{tabular}{lcccccc}
\toprule
  Factor  & \!Scenarios\! & \!(S,\,Colless)\! & \!(S,\,cherries)\!
  &\,(S,\,$\sigma_N^2$)\, &\,(S,\,$B_1$)\, &\,(S,\,$B_2$)\, \\
\midrule
      ERM  &     17   &   4.62    &   2.94   &  \HIGH  2.35 &  2.59  & 2.50 \\ 
      PDA  &     17   &   4.29    &   2.82   &    2.82 & \HIGH 1.35  & 3.71 \\
       AB  &     17   &   4.12    &   3.03   &  \HIGH  2.21 &  2.65  & 3.00 \\
     MFL1  &     17   &   4.53    &   3.59   &    2.12 &  2.70  & \HIGH 2.06 \\
     MFL2  &     17   &   3.94    &   3.35   &  \HIGH   2.21 &  2.79  & 2.70 \\
\midrule
 TreeBASE\!\!\!\!  &     25   &   4.70    &   3.48   &    2.38 &  2.24  & \HIGH 2.20 \\
 OrthoMaM\!\!\!\!  &     10   &   2.80    &   3.80   &   \HIGH  1.80 &  3.30  & 3.30 \\
PhylomeDB\!\!\!\!  &     25   &   4.76    &   3.26   &    1.94 &  3.30  & \HIGH 1.74 \\
  HOGENOM\!\!\!\!  &     25   &   4.04    &   2.44   &    2.92 &  \HIGH 1.36  & 4.24 \\
\midrule
$n\in$  6--12  & 15   &   4.63    &   3.53   &  \HIGH   1.47 &  2.40  & 2.97 \\
$n\in$ 12--24  & 15   &   4.53    &   3.67   &    2.27 &  2.60  & \HIGH 1.93 \\
$n\in$ 25--50  & 15   &   4.40    &   2.73   &    2.87 &  \HIGH 2.33  & 2.67 \\
$n\in$ 50--100\!\!\!\! & 20   &   4.15    &   3.02   &  \HIGH   2.40 &  2.42  & 3.00 \\
$n > 100$      & 20   &   3.95    &   2.90   &   2.60 & \HIGH   2.35  & 3.20 \\
\midrule
   All         & 85   &   4.30    &   3.15   &   \HIGH  2.34 &  2.42  & 2.79 \\
\bottomrule
\end{tabular}
\caption{Average ranks of the pairs of balance indices that include the
  Sackin index.}
\label{tabResultsBivariateSackin}
\end{table}

One clear conclusion emerges from Tables~\ref{tabResultsBivariateColless}
and~\ref{tabResultsBivariateSackin}: that the Sackin index is the
worst possible balance index to complement the Colless index, and vice versa.
In fact, even complementing them with the number of cherries (which was the
only index to perform significantly worse than the others when considered
alone) almost invariably gives better results.
As mentioned above, this result is not surprising given the strong correlation
of the Sackin and Colless indices.
It is nevertheless of significant importance, considering how these indices
are used in practice.

Another, more unexpected conclusion to be drawn from
Tables~\ref{tabResultsBivariateColless} and~\ref{tabResultsBivariateSackin}
is that $\sigma_N^2$ seems to be consistently better than any
other index at complementing the Colless index (with the possible exception of
the $B_1$ index).  This also seems to be the case when it comes to
complementing the Sackin index, even though in that case things are a bit more
nuanced and both the $B_1$ and $B_2$ index can be useful.

\pagebreak

Lastly, although this is not the aim of this study, having
at our disposal a measure of the power of every pair of
balance indices also makes it tempting to adopt a  more exploratory
approach and assess whether some more exotic combinations of balance indices
might be useful. It turns out that this is the case: out of the 15
pairs of indices (expected rank under uniformity:~8),
{$(B_1$, $B_2)$} and {$(\sigma_N^2, B_2)$} stand out
for having an average rank that is slightly better than other pairs.
Moreover, the fact that $B_2$ is part of these seemingly optimal pairs
of balance indices reinforces the idea that its relevance may have been
underestimated. Finally, let us point out that, strikingly,
{(Sackin, Colless)} turns out to have the worst performance of all possible
pairs of indices.

\vspace{1ex} %HARDCODED

\begin{table}[h!]
\centering
\captionsetup{width=0.7\linewidth}
\begin{tabular}{r|cccccc}
\toprule
  & Sackin & Colless & Cherries & $\sigma_N^2$ & $B_1$ & $B_2$ \\
\midrule
  Sackin     & --- & 11.44 & 8.77 & 6.52 & 6.26 & 7.59 \\
  Colless    & & --- &  9.02 & 6.71 & 7.26 & 8.76  \\
  Cherries   & & & --- & 9.91 & 11.12& 7.01  \\
$\sigma_N^2$ & & & & --- & 7.93 & \HIGH 6.08 \\
$B_1$        & & & & & --- & \HIGH 5.63 \\
\bottomrule
\end{tabular}
  \caption{Average rank (over all 85 scenarios) of each pair of balance indices.
  The two best performing pairs are highlighted.}
\label{tabOverallBivariate}
\end{table}

\section{Concluding comments} \label{secConclusion}

Given its intuitive definition, it is legitimate to wonder why $B_2$ is not
more prominent in phylogenetics; or why it has in fact not been studied
in its own right before. To conclude this article, we speculate as to why
this might be the case.

One first reason could be that the intuition behind the definition
of $B_2$ is not always well understood. Indeed, although its probabilistic
interpretation is very clearly laid out in Shao and Sokal's original paper,
it is almost never mentioned in subsequent works that use $B_2$.
For instance, to motivate its definition Kirkpatrick and Slatkin
merely say that \emph{``The statistic $B_2$ was suggested by the
Shannon--Wiener statistic. The Shannon--Wiener statistic was developed as an
index of information content, and so a measure of tree shape related to it
might be a useful statistic for detecting patterns.''} \cite{Kirkpatrick1993}.
In most sources, $B_2$ is only described, somewhat vaguely, either as
an information theoretic measure of balance or as a weighted variant of
the Sackin index. In our opinion this is unfortunate, because the
probabilistic interpretation of $B_2$ is precisely what sets it apart
from other balance indices. For instance, it has already been pointed out
in the literature that it is not entirely clear why the Sackin index happens to
correlate with our intuition of phylogenetic (im)balance -- as illustrated by
the fact that Sackin's original idea had little to do with the index that came
to bear his name~\cite{Coronado2020b}.

A second possible reason for the lack of popularity of $B_2$ is that it
might be perceived as a mathematically unwieldy quantity. For instance, in
the only mention of $B_2$ that we could find in the mathematical literature \cite{Blum2006a},
after studying the asymptotic distribution of the Colless and Sackin indices 
under the ERM model \revision{Blum, François and Janson} conclude by saying that
\emph{``In the same spirit, we believe that the $B_1$ index of Shao and Sokal
could be studied without difficulties. Studying the remaining statistics
($B_2$ and $\sigma_N^2$) would nevertheless require considerably more
effort.''} The results of Sections~\ref{secExtrem} and~\ref{secRandomTrees},
as well as the simplicity of their proofs, show that the idea that $B_2$ is
untractable is unjustified. In fact, in some respects $B_2$ seems
\emph{more} tractable than other classic balance indices
(compare for instance its expected value under the ERM model to that of the
Colless index~\cite{Heard1992}, and consider the fact that the expected value
of the Colless index under the PDA model is currently not known).

Finally, the third -- and most likely main -- reason for the current status
of $B_2$ is probably the reputation of being less useful than other
balance indices that it earned from Agapow and Purvis's 2002 study, which they
conclude by saying: \emph{``$B_2$ never performs well and should not be used.''}
\cite{Agapow2002}. While this conclusion is justified in their particular
setting, it does not hold when their hypotheses are relaxed -- in
particular
when using other null models than the ERM model, or when working with some
specific types of phylogenies that occur in the real world. In
fact, that $B_2$ could be useful in some contexts could already be observed in
some other studies comparing the performance of various balance indices -- in
particular in~\cite{Maia2004}, where $B_2$ is found to be the second most
powerful statistic considered, beating the Colless index. The study by Agapow
and Purvis also does not take into account the possibility of using balance
indices jointly. As shown in Section~\ref{secStats}, this completely changes
the relevance of the indices, as some that perform well on
their own, such as the Sackin and Colless indices, can perform very poorly when
used jointly.

\enlargethispage{1ex} %HARDCODED

In conclusion, none of the reasons that currently make $B_2$ a
``second-rate'' balance index seems justified. Our work
calls for a reevaluation of the status of $B_2$ in phylogenetics
-- in particular in the age of phylogenetic networks, where alternatives
will have to be found to the classical
measures of phylogenetic balance that are used on trees.

\section*{Acknowledgements}

The authors thank Simon Penel for his help with the HOGENOM database
and Roberto Bacilieri for helpful discussions.

FB and CS were funded by grants ANR-16-CE27-0013 and
ANR-19-CE45-0012 from the
Agence Nationale de la Recherche, respectively;
GC was funded by FEDER\,/\,Ministerio de Ciencia, Innovación y
Universidades\,/\,Agencia Estatal de Investigación project
PGC2018-096956-B-C43.

%\pagebreak %HARDCODED

\phantomsection
\addcontentsline{toc}{section}{References}
\bibliographystyle{abbrvnat}
\bibliography{biblio} 

\appendix
\setcounter{equation}{0}
\renewcommand{\thesection}{A}
\renewcommand{\thesubsection}{A.\arabic{subsection}}
\renewcommand{\theequation}{A.\arabic{equation}}
\renewcommand{\thetheorem}{\thesubsection.\arabic{theorem}}

\pagebreak

\section*{\centering Appendix}
\addcontentsline{toc}{section}{Appendix}

\subsection{Variance of $B_2$ in binary Galton--Watson trees}%
\label{appProofVarBinaryGW}

In this section, we prove Proposition~\ref{propVarBinaryGW} concerning
the variance of $B_2$ in binary Galton--Watson trees. Let us start with
a standard result, which we recall and prove for the sake of completeness.

\begin{lemma} \label{lemmaSecondMomentGW}
Let $(Z_t)$ be a Galton--Watson process with offspring distribution
$Y \sim 2\, \mathrm{Bernoulli(1/2)}$. Then, $\Expec{Z_t^2} = t + 1$.
\end{lemma}

\begin{proof}
Since $Z_{t + 1} = \sum_{i = 1}^{Z_t} Y_t(i)$, we have
\[
  Z_{t + 1}^2 \;=\; \sum_{i = 1}^{Z_t} Y_t(i)^2 \;+\;
  \sum_{i = 1}^{Z_t} \sum_{j \neq i}^{Z_t} Y_t(i)\, Y_t(j)\,.
\]
Letting $\mathcal{F}_t$ be the natural filtration of the process,
we thus have
\[
  \Expec{Z_{t + 1}^2 \given \mathcal{F}_t} \;=\;
  \Expec*{\normalsize}{Y^2} \,Z_t \;+\;
  \Expec{Y}^2 \, Z_t(Z_t - 1)
\]
and, as a result,
\[
  \Expec{Z_{t + 1}^2} \;=\;
  \Expec*{\normalsize}{Y^2} \,\Expec{Z_t} \;+\;
  \Expec{Y}^2 \, \Expec*{\normalsize}{Z_t^2} \;-\;
  \Expec{Y}^2 \, \Expec{Z_t}\,.
\]
Since here $\Expec{Y} = 1$, $\Expec{Y^2} = 2$ and $\Expec{Z_t} = 1$, this
simplifies to
\[
  \Expec{Z_{t + 1}^2} \;=\;
  \Expec*{\normalsize}{Z_t^2} \;+\; 1\,.
\]
The lemma then follows by induction, since $\Expec{Z_0^2} = 1$.
\end{proof}

Let us now recall Proposition~\ref{propVarBinaryGW} and prove it.

\begin{repproposition}{propVarBinaryGW}
Let $T$\! be a critical binary Galton--Watson tree, that is, assume that
the offspring distribution is $\Prob{Y = 0} = \Prob{Y = 2} = 1/2$. Then,
\[
  \Var{B_2(\Ball{T}{t})} \;=\;
  \frac{4}{3} - 2^{-t+2} + 4^{-t}\mleft(t + \frac{8}{3}\mright)\,.
\]
\end{repproposition}

\begin{proof}
To alleviate the notation, let us write
$B_t = B_2(\Ball{T}{t})$. With this notation, since 
in the case of a binary Galton--Watson tree
$\Indic{Y_t > 0} = Y_t / 2$,
Equation~\eqref{eqProofThmExpecGW} from the proof
of Theorem~\ref{thmExpecGW} becomes
\begin{equation} \label{eqProofPropVarBinaryGW01}
  B_{t + 1} \;=\; B_t \;+\; 2^{-(t + 1)} \sum_{i = 1}^{Z_t} Y_t(i) \,.
\end{equation}
Therefore, letting $\mathcal{F}_t$ denote the natural filtration of the process
and using that $\Expec{Y} = 1$, we have
\[
  \Expec{B_{t + 1}^2 \given \mathcal{F}_t} \;=\;
  B_t^2 \;+\; 2^{-t} B_t\, Z_t \;+\;
  2^{-2(t + 1)} \, \Expec{Z_{t + 1}^2 \given \mathcal{F}_t}.
\]
As a result,
\begin{equation}\label{eqProofPropVarBinaryGW02}
  \Expec*{\normalsize}{B_{t + 1}^2} \;=\;
  \Expec*{\normalsize}{B_t^2\,} \;+\; 2^{-t} \,\Expec{B_t Z_t} \;+\;
  2^{-2(t + 1)} \, \Expec*{\normalsize}{Z_{t + 1}^2}.
\end{equation}
Let us now turn our attention to $\Expec{B_t Z_t}$. Using
Equation~\eqref{eqProofPropVarBinaryGW01}, we get
\begin{align*}
  B_{t+1} Z_{t + 1}
  \;&=\; \mleft(B_t + 2^{-(t + 1)} Z_{t + 1}\mright) Z_{t+1} \\
  \;&=\; B_t \sum_{i = 1}^{Z_t} Y_t(i) \;+\; 2^{-(t + 1)} Z_{t + 1}^2 \,, 
\end{align*}
and since
$\Expec*{\normalsize}{B_t \sum_{i = 1}^{Z_t} Y_t(i) \given \mathcal{F}_t} =
  B_t Z_t$ this yields
\[
  \Expec{B_{t+1} Z_{t + 1}}
  \;=\; \Expec{B_t Z_t} \;+\; 2^{-(t + 1)} \Expec*{\normalsize}{Z_{t + 1}^2} \,.
\]
By Lemma~\ref{lemmaSecondMomentGW}, $\Expec*{\normalsize}{Z_{t + 1}^2} = t + 2$.
Since $\Expec{B_0 Z_0} = 0$, we thus have
\[
  \Expec{B_t Z_t} \;=\; \sum_{s = 0}^{t - 1} 2^{-(s + 1)} (s + 2) \;=\;
  3 - 2^{-t}(t + 3)\,.
\]
Plugging this and $\Expec*{\normalsize}{Z_{t + 1}^2} = t + 2$ in
Equation~\eqref{eqProofPropVarBinaryGW02}, we get the following closed
recurrence relation for $\Expec{B_t^2}$:
\[
  \Expec*{\normalsize}{B_{t + 1}^2} \;=\;
  \Expec*{\normalsize}{B_t^2\,} \;+\; 2^{-t} \,\big(3 - 2^{-t}(t + 3)\big) \;+\;
  2^{-2(t + 1)} \, (t + 2).
\]
Solving this recurrence relation with the initial condition
$\Expec*{\normalsize}{B_0^2} = 0$ then yields
\[
  \Expec*{\normalsize}{B_t^2\,} \;=\;
  \frac{7}{3} - 6\cdot2^{-t} + 4^{-t}\mleft(t + \frac{11}{3}\mright)\,.
\]
Finally, since by Theorem~\ref{thmExpecGW}, $\Expec{B_t} = 1 - 2^{-t}$, we have
\[
  \Var{B_t} \;=\;
  \frac{4}{3} - 4\cdot2^{-t} + 4^{-t}\mleft(t + \frac{8}{3}\mright)\,,
\]
concluding the proof.
\end{proof}

\subsection{Bounds on the number of distinct values of $B_2$}
\label{appDistinctValuesB2}

\begin{proposition} \label{propDistinctValuesB2}
Let $\mathscr{T}_n$ denote the set of rooted binary trees with $n$ leaves,
labeled or unlabeled. Then,
\[
  2^{\Floor{n/2} - 1} \;\leq\; \# B_2(\mathscr{T}_n) \;\leq\; a_n
\]
where $a_n$ is sequence \href{https://oeis.org/A002572}{A002572} in the Online
Encyclopedia of Integer Sequences~\cite{OEIS} and satisfies
$a_n \sim K \rho^n$, with $\rho \approx 1.7941$
and $K \approx 0.2545$ the Flajolet--Prodinger constant.
\end{proposition}

\begin{proof}
The upper bound is obtained by noting that $B_2(T)$ is a function of the
multiset of the depths of the leaves of the binary tree $T$. Therefore,
$B_2$ cannot take more values than the number $a_n$ of such multisets, whose
asymptotics were characterized by Flajolet and Prodinger in \cite{Flajolet1987}.

\pagebreak
To obtain the lower bound, we exhibit $2^{\Floor{n/2} - 1}$ rooted binary trees
with $n$ leaves whose $B_2$ indices are different. For any integer~$m$
and any $\mathbf{x} \in \Set{0, 1}^m$, let $T(\mathbf{x})$ denote the
ordered (that is, embedded in the plane) rooted binary tree obtained by the
following sequential construction: starting from the binary tree with two
leaves, for $k = 1, \ldots, m$,
\begin{itemize}
  \item If $x_k = 0$, graft a cherry on the left-most leaf with depth~$k$.
  \item If $x_k = 1$, graft a cherry on each of the two left-most leaves
    with depth~$k$.
\end{itemize}
This construction is illustrated in Figure~\ref{figDistinctB2}.

\begin{figure}[h!]
  \centering
  \captionsetup{width=0.95\linewidth}
  \includegraphics[width=0.95\linewidth]{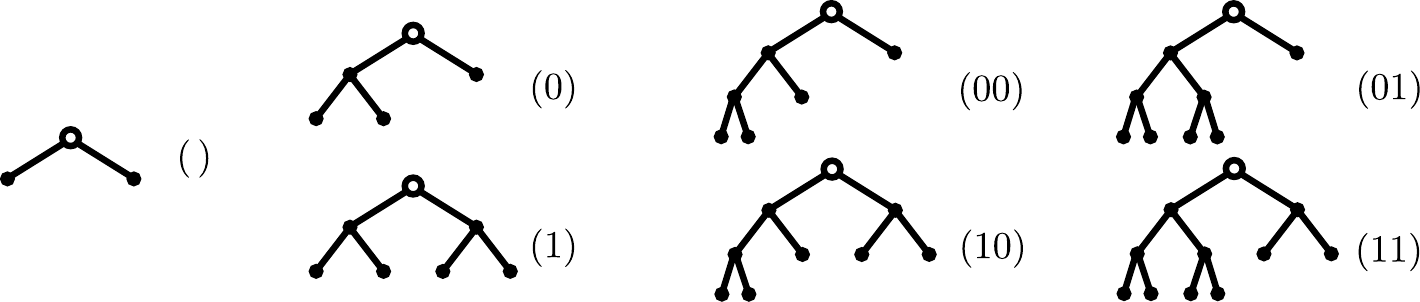}
  \caption{The trees $T(\mathbf{x})$ for $m \in \Set{0, 1, 2}$. Each tree
  is represented with the corresponding vector $\mathbf{x} \in \Set{0, 1}^m$ on the right.} 
\label{figDistinctB2}
\end{figure}

Clearly, $T(\mathbf{x})$ has $2 + \sum_{k = 1}^{m} (x_k + 1)$ leaves and, by
Corollary~\ref{corGraftingCherry},
\[
  B_2(T(\mathbf{x})) \;=\; 1 \;+\; \sum_{k = 1}^{m} (x_k + 1)\, 2^{-k} \,.
\]
As a result,
\[
  \Set*[\big]{B_2(T(\mathbf{x})) \suchthat \mathbf{x} \in \Set{0, 1}^m}
  \;=\; \Set{u_m + i\, 2^{-m} : i = 0, \ldots, 2^m - 1}
\]
(to see this, note that $B_2(T(\mathbf{x})) = u_m + \underline{\mathbf{x}}_{(2)}$,
where $\underline{\mathbf{x}}_{(2)} = \sum_{k} x_k 2^{-k}$ denotes
the dyadic rational of $\COInterval{0, 1}$ whose binary expansion is $\mathbf{x}$).
Thus, this construction generates $2^m$ trees whose
$B_2$ indices differ by at least $2^{-m}$. However, these trees do not have
the same number of leaves.

Let us first assume that $n$ is even. Then,
with $m = n/2 - 1$, $T(1\cdots1)$ has~$n$ leaves and every other tree
$T(\mathbf{x})$ with $\mathbf{x} \in \Set{0, 1}^m$ has:
\begin{mathlist}
\item $n - k_\mathbf{x}$ leaves, with $1 \leq k_\mathbf{x} \leq m$;
\item its left-most leaf at depth $m + 1$.
\end{mathlist}
Now, if we graft a caterpillar with $k_\mathbf{x} + 1$ leaves on the left-most
leaf at depth $m + 1$ and let $T'(\mathbf{x})$
denote the resulting tree, then:
\begin{mathlist}
\item[($\mathrm{i}^\prime$)] $T'(\mathbf{x})$ has $n$ leaves;
\item[($\mathrm{ii}^\prime$)] by Proposition~\ref{propGrafting},
  $B_2(T'(\mathbf{x})) - B_2(T(\mathbf{x})) = 2^{-(m+1)}(2 - 2^{-k+1}) \in
  \OpenInterval{0, 2^{-m}}$.
\end{mathlist}
Since the $B_2$ indices of the trees $T(\mathbf{x})$ differ by at
least $2^{-m}$, by point~($\mathrm{ii}^\prime$) the $B_2$ indices of
the trees $T'(\mathbf{x})$ are all different, thereby
proving the proposition in the case where $n$ is even.
 
If $n$ is odd, do the same construction, again with
$m = \lfloor n / 2\rfloor - 1$, to get $2^m$ trees with $n - 1$ leaves. Then, for
each of these trees, graft a cherry on the sibling of the left-most vertex at
depth $m + 1$ (which exists and is always a leaf). This extra step increases
$B_2$ by the same amount $2^{-(m + 1)}$ for every tree, concluding the proof.
\end{proof}

\end{document}